\RequirePackage{fix-cm}
\documentclass[letter, 11pt]{svjour3}                     
\smartqed  

\usepackage{natbib}
\usepackage{latexsym}
\usepackage{graphics}
\usepackage{amsfonts}
\usepackage{enumerate}
\usepackage{graphicx}
\usepackage{verbatim}
\usepackage{subfig}
\usepackage{caption}
\usepackage{footnote}
\usepackage{slashbox}
\usepackage{lscape}
\usepackage{hhline}
\usepackage{array}
\usepackage{graphicx}
\usepackage{setspace}
\usepackage{amssymb}
\usepackage{amsmath}
\usepackage{amsfonts}
\usepackage{times}
\usepackage{color}
\usepackage{indentfirst}
\usepackage{epsfig}
\usepackage{lscape,ctable,rotating}
\usepackage{enumitem}
\setlist{nolistsep}

\usepackage{epstopdf}
\usepackage{color}
\newtheorem{rem}{Remark}[section]

\newcommand{\R}{\mathbb{R}}

\newcommand{\levy}{{L\'evy}}

\newcommand{\var}{{\textup{var}}}

\begin{document}

\title{First Passage Time for Tempered Stable Process and Its Application to Perpetual American Option and Barrier Option Pricing}

\titlerunning{First Passage Time for TS Process \& Its Application}        

\author{Young Shin Kim}


\institute{Young Shin Kim \at
              College of Business, Stony Brook University, New York, USA\\
              \email{aaron.kim@stonybrook.edu}  
}

\date{Received: Sep 22 / Accepted: date}

\maketitle

\begin{abstract}
In this paper, we will discuss an approximation of the characteristic function of the first passage time for a \levy~process using the martingale approach. The characteristic function of the first passage time of the tempered stable process is provided explicitly or by an indirect numerical method. This will be applied to the perpetual American option pricing and the barrier option pricing. Numerical illustrations are provided for the calibrated parameters using the market call and put prices.
\end{abstract}

\section{Introduction}
Since \cite{BlackScholes:1973} have introduced the no-arbitrage option pricing model and formula, the Black-Scholes (BS) model became the most popular model in finance. Moreover, pricing path dependent exotic options, including American perpetual option and barrier options, is an important topic in finance. 
After the Black-Monday stock market crash in 1987, the volatility smile effect in option market have been observed and many scientists have introduced many advanced models to describe the smile effect. The \levy~process option pricing models (or simply \levy~models), based on tempered stable (TS) processes, including normal tempered stable (NTS) and CGMY processes, are popular models to explain the volatility smile effect in European call and put prices (See \cite{BarndorffNielsenLevendorskii:2001}, \cite{BarndorffNielsenShephard:2001}, \cite{CGMY:2002}, \cite{Boyarchenko_Levendorskii:2000}, \cite{Koponen:1995}, and \cite{RachevKimBianchiFabozzi:2011a}).

The distribution of the first passage time for the arithmetic Brownian motion is an essential topic to calculate those path dependent option prices. We have the distribution of the first passage time of Brownian motion in literature including \cite{Barndorff-Nielsen401}.
Sequentially the distribution of first passage time on \levy~process has been studied by \cite{hurd2009}, \cite{rogers2000} and others.

The perpetual American option and barrier option pricing also not only studied on BS model but also studied on \levy~models. \cite{GerberShiu:1994} discussed the perpetual American option pricing formula on the BS model with the martingale approach, and \cite{Boyarchenko_Levendorskii:2002a} found perpetual American option pricing formula using the Wiener–Hopf factorization on the \levy~model. The barrier option price formula under the BS model is provided in literature including \cite{Hull:2015}.
Additionally, \cite{Boyarchenko_Levendorskii:2002b} presented the barrier option pricing method on \levy~model.
The partial integro-differential equation method has been very popularly used for barrier option pricing (See \cite{math4010002} and \cite{Cont_Tankov:2004}) on \levy~model.
Recently, \cite{2014arXiv1403.1816B} discussed barrier option pricing method using A transform.

In this paper, we will discuss the characteristic function of the first passage time of a subclass of \levy~process containing Brownian Motion and TS processes (NTS or CGMY process) using the martingale approach. Since the martingale approach does not works for the process with jumps, we will use a continuous approximation of the \levy~process. After then we find an approximation form of the characteristic function for the first passage time of \levy~process. In some special cases, we will see the closed form of characteristic function. If the closed form solution is not allowed then the numerical method can be used to find it. The characteristic function of the first passage time will be applied to find perpetual American option prices and barrier option prices. The numerical methods and performance for pricing perpetual American option and barrier option will be discussed with empirical market data.

The remainder of this paper is organized as follows. The characteristic function of the first passage time for some \levy~process using the martingale method is deduced in Section 2.  The approximation case for the \levy~process with jumps also discussed in the section.
The perpetual American option pricing and the barrier option pricing is discussed in Section 3, together with numerical illustrations. Section 3 summarizes the main findings. In the appendix, we explain how pricing formulas of perpetual American option and barrier call and put options are obtained.

\section{Characteristic function of the first passage time}\label{section2}
Let $X = (X(t))_{t\ge0}$ be a \levy~process.
Suppose that $\phi_{X(t)}$ is the characteristic function (ch.F) of $X(t)$ and $\psi_{X}$ is the \levy~symbol of $X$ that is given by $\psi_{X}(u) = \log\phi_{X(1)}(u)$ so that $\phi_{X(t)}(u)=e^{t\psi_{X}(u)}$ (see \cite{Applebaum:2004}). Let $l\in\R$ be a level.
We define a first passage time $\tau(l)$ for the \levy~process $X$ to touch the level $l$ as follows:
\begin{equation}\label{eq:tau(l)}
\tau(l) = \begin{cases} \inf \{t\ge0 | X(t)\le l \} & \text{ if } l<0 \\
 \inf \{t\ge0 | X(t)\ge l \} & \text{ if } l>0 
\end{cases}.
\end{equation}
\begin{lemma}\label{Lemma:ChfTau}
Suppose $X$ is a continuous \levy~process.
For all $u\in\R$, if there exist a complex function $\eta(u)$ such that $\Re\left(-l\eta(u)\right)<0$, $\phi_{X(1)}(-i\eta(u))$ is well defined, and
\begin{equation}\label{eq:martingalecondition2}
iu+\psi_{X}(-i\eta(u))=0 
\end{equation}
then the ch.F of $\tau(l)$ equals to
\begin{equation}
\label{eq:chf_tau}
\phi_{\tau(l)}(u) = E\left[e^{iu \tau(l)} \right]=e^{-l\eta(u)}.
\end{equation}
\end{lemma}
\begin{proof}
For given $u\in\R$, if there exist $\eta(u)$ satisfying (\ref{eq:martingalecondition2}), then we have
\begin{equation}\label{eq:martingalecondition}
1=e^{t(iu+\psi_{X}(-i\eta(u)))} = E\left[e^{iu t + \eta(u) X(t)}\right],
\end{equation}
since $\log E[e^{ixX(t)}]=t\psi_{X}(x) $.
We have
\[
1 = E\left[e^{iu \tau(l) + \eta(u) X(\tau(l))}\right] 
= E\left[e^{iu \tau(l)} e^{l\eta(u)}\right]
= e^{l\eta(u)} E\left[e^{iu \tau(l)} \right],
\]
and we obtain ch.F of $\tau(l)$ as $E\left[e^{iu \tau(l)} \right]=e^{-l\eta(u)}$.
Since we have 
\[
|e^{-l\eta(u)}|=|E[e^{iu\tau(l)}]|\le E[|e^{iu\tau(l)}|]=1,
\]
and 
\[
e^{-l\eta(u)} = e^{\Re(-l\eta(u))}e^{i\Im(-l\eta(u))},
\]
we obtain the condition $\Re(-l\eta(u))\le0$. Moreover, since we have 
\[
E\left[e^{\eta(u) X(t)}\right] = \exp\left(t \psi_{X}(-i\eta(u))\right),
\]
the equation (\ref{eq:martingalecondition}) holds if and only if
$iu+\psi_{X}(-i\eta(u))=0$. 
\end{proof}
\begin{rem}
Applying Lemma \ref{Lemma:ChfTau}, we can obtain the probability density function (pdf) of $\tau(l)$  by the inverse Fourier transform as follows:
\begin{equation}\label{eq:pdf_FFT}
f_{\tau(l)}(x)=\frac{1}{2\pi}\int_{-\infty}^\infty e^{-iux}\phi_{\tau(l)}(u) du = \frac{1}{2\pi}\int_{-\infty}^\infty e^{-iux-l\eta(u)} du.
\end{equation}
\end{rem}
\begin{example}[\emph{Brownian Motion}] Let $X = (X(t))_{t\ge0}=(\mu t+\sigma B(t))_{t\ge0}$ where $(B(t))_{t\ge0}$ is Brownian motion, $\mu\in\R$, and $\sigma>0$.
Since the ch.F of $X(t)$ is 
\[
\phi_{X(t)}(u)=\exp\left(\mu  iut -t\frac{\sigma^2u^2}{2}\right)
\]
The equation (\ref{eq:martingalecondition2}) is equal to
\[
iu  + \mu \eta(u) + \frac{\sigma^2\eta(u)^2}{2}=0,
\]
and $\eta(u)$ satisfying the equation is
\[
\eta(u)=\frac{-\mu \pm \sqrt{\mu^2-2\sigma^2ui}}{\sigma^2}.
\]
For the condition $\Re\left(-l\eta(u)\right)<0$, 
we have  
\[
\eta(u)=\begin{cases}
\frac{-\mu + \sqrt{\mu^2-2\sigma^2ui}}{\sigma^2}, & \text{ if } l>0\\
\frac{-\mu - \sqrt{\mu^2-2\sigma^2ui}}{\sigma^2}, & \text{ if } l<0
\end{cases}.
\]
Hence the ch.F of the first passage time $\tau(l)$ for $X$ is
\[
\phi_{\tau(l)}(u) = \begin{cases} \exp\left(\frac{l\mu - l\sqrt{\mu^2-2\sigma^2ui}}{\sigma^2}\right) &\text{ if } l>0 \\
\exp\left(\frac{l\mu + l\sqrt{\mu^2-2\sigma^2ui}}{\sigma^2}\right) &\text{ if } l<0 
\end{cases},
\]
which is well known ch.F of inverse Gaussian distribution.
\end{example}

We cannot use Lemma \ref{Lemma:ChfTau} for the \levy~process with jumps, since $X(\tau(l)) = l$ is not true in general. To escape the problem, we define continuous approximation process $X^c=(X^c(t))_{t\ge0}$ for the \levy~process $X$ as
\[
X^c(t) = \begin{cases}
X(t) &\text{ if } t \in P\\
\displaystyle \frac{(t-t_i)X(t_{i+1})+(t_{i+1}-t)X(t_i)}{t_{i+1}-t_{i}} &\text{ if } t_i<t<t_{i+1} \text{ for } t_i \in P
\end{cases}
\]
where $P$ is a partition of a real interval $[0,T]$ as 
\[
P = \{ x_i | 0=x_1<x_2<\cdots<x_n< \cdots \}.
\]
We refer to $X^c$ as the \emph{continuously approximated process} of $X$.We have characteristic function of $X^c(t)$ as follows:
\[
\phi_{X^c(t)}(u) = \exp\left(t_i\psi_X(u)+(t_{i+1}-t_i)\psi_X(\Delta u)\right) \text{ for } t\in[t_i,t_{i+1}),
\]
where 
\[
\Delta = \frac{t-t_i}{t_{i+1}-t_i}.
\]
We use numerical approximation of $\psi_X(\Delta u) \approx \Delta \psi_X(u)$, then we have
\[
\phi_{X^c(t)}(u) \approx  \exp\left(t_i\psi_X(u)+(t_{i+1}-t_i)\Delta \psi_X(u)\right) = \phi_{X(t)}(u)
\]
After then we use Lemma \ref{Lemma:ChfTau} for $X^c$ with approximated ch.F of $X^c$.
That is if $\eta(u)$ satisfy \eqref{eq:martingalecondition2}, then the first passage time $\tau(l)$ of $X^c$ has an approximation of the characteristic function as
\[
\phi_{\tau(l)}(u) \approx e^{-l\eta(u)}.
\]

\subsection{Cases of NTS Process and Normal Inverse Gaussian Process}
Let $\alpha\in(0,2)$, $\theta>0$, $\beta \in\R$, $\gamma>0$ and $mu\in\R$.
Consider a pure jump \levy~process $X = (X(t))_{t\ge0}$ whose ch.F $\phi_{X(t)}$ is equal to
\[
\phi_{X(t)}(u)  =\exp\left( (\mu-\beta) iut -\frac{2\theta^{1-\frac{\alpha}{2}}}{\alpha}
t\left(\left(\theta-i\beta u+\frac{\gamma^2 u^2}{2}\right)^{\frac{\alpha}{2}}-\theta^{\frac{\alpha}{2}}\right)\right).
\]
The process $X$ is referred to as the the NTS process with parameters $(\alpha$, $\theta$, $\beta$, $\gamma$, $\mu)$ and denoted by $X\sim \textup{NTS}(\alpha$, $\theta$, $\beta$, $\gamma$, $\mu)$. The NTS process has finite exponential moments for a closed interval, That is, $E[e^{aX(t)}]<\infty$
if 
\[
a\in\left[\frac{1}{\gamma^2}\left(-\beta-\sqrt{\beta^2+2\gamma^2\theta}\right), \frac{1}{\gamma^2}\left(-\beta+\sqrt{\beta^2+2\gamma^2\theta}\right)\right].
\]

If $X\sim \textup{NTS}(\alpha$, $\theta$, $\beta$, $\gamma$, $\mu)$ with $\alpha = 1$, then the process $X$ is referred to as the the normal inverse Gaussian (NIG) process with parameters $(\theta$, $\beta$, $\gamma$, $\mu)$ and denoted by $X\sim \textup{NIG}(\theta$, $\beta$, $\gamma$, $\mu)$. The ch.F of the NIG process $X$ is equal to
\[
\phi_{X(t)}(u)  =\exp\left( (\mu-\beta) iut+2t\theta - 2t\theta^{\frac{1}{2}}
\left(\theta-i\beta u+\frac{\gamma^2 u^2}{2}\right)^{\frac{1}{2}}\right).
\]

If $X\sim \textup{NTS}(\alpha$, $\theta$, $\beta$, $\gamma$, $\mu)$ where $\mu = 0$ and $\gamma=\sqrt{1-\beta^2 \left(\frac{2-\alpha}{2\theta}\right)}$ with $-\sqrt{ \frac{2\theta}{2-\alpha}}<\beta<\sqrt{ \frac{2\theta}{2-\alpha}}$ then $E[X(t)]=0$ and $\var(X(t))=t$ for all $t>0$.
In this case, the process $X$ is referred to as the \textit{standard NTS process} with parameters $(\alpha$, $\theta$, $\beta)$ and denoted by $X\sim \textup{stdNTS}(\alpha$, $\theta$, $\beta)$. With the same argument,  $X\sim \textup{NIG}(\theta$, $\beta$, $\gamma$, $\mu)$ where $\mu = 0$ and $\gamma=\sqrt{1-\frac{\beta^2}{2\theta}}$ with $-\sqrt{ 2\theta}<\beta<\sqrt{2\theta}$,  then $E[X(t)]=0$ and $\var(X(t))=t$ for all $t>0$.
In this case, the process $X$ is referred to as the \textit{standard NIG process} with parameters $(\theta$, $\beta)$ and denoted by $X\sim \textup{stdNIG}(\theta$, $\beta)$.

By applying Lemma \ref{Lemma:ChfTau} to the process $X\sim \textup{NIG}(\theta$, $\beta$, $\gamma$, $\mu)$, we find $\eta(u)$ satisfying (\ref{eq:martingalecondition2}) that is
\[
0=\left(2\theta+iu\right) +(\mu-\beta)\eta(u) -2\theta^{\frac{1}{2}}
\left(\theta-\beta \eta(u)-\frac{\gamma^2 \eta(u)^2}{2}\right)^{\frac{1}{2}}
\]
or
\[
\left((\mu-\beta)^2+2\theta\gamma^2\right) \eta(u)^2 + 2\left(2\mu\theta+(\mu-\beta) iu\right) \eta(u) - u^2+4\theta iu = 0.
\]
Finally, we obtain the solution 
\[
\eta(u) = \frac{-\left(2\mu\theta+(\mu-\beta)iu\right)\pm\sqrt{\left(2\mu\theta+(\mu-\beta) iu\right)^2+\left((\mu-\beta)^2+2\theta\gamma^2\right)\left(u^2-4\theta iu\right)}}{(\mu-\beta)^2+2\theta\gamma^2}.
\]

Since NIG process is a pure jump \levy~process, we cannot use Lemma \ref{Lemma:ChfTau} directly. Instead, we use a continuously approximated process $X^c$ of the process $X$ for a partition $P$. Satisfying the condition $\Re\left(-l\eta(u)\right)<0$ for all $u$, we have the ch.F of the first passage time $\tau(l)$ for process $X^c$ as
\[
\phi_{\tau(l)}(u) \approx \begin{cases} \exp\left(-l\eta^+(u)\right) &\text{ if } l>0 \\
\exp\left(-l\eta^-(u)\right) &\text{ if } l<0 
\end{cases}
\]
where
\[
\eta^+(u) = \frac{-\left(2\mu\theta+(\mu-\beta) iu\right)+\sqrt{\left(2\mu\theta+(\mu-\beta) iu\right)^2+\left((\mu-\beta)^2+2\theta\gamma^2\right)\left(u^2-4\theta iu\right)}}{(\mu-\beta)^2+2\theta\gamma^2}, 
\]
and 
\[
\eta^-(u) = \frac{-\left(2\mu\theta+(\mu-\beta) iu\right)-\sqrt{\left(2\mu\theta+(\mu-\beta) iu\right)^2+\left((\mu-\beta)^2+2\theta\gamma^2\right)\left(u^2-4\theta iu\right)}}{(\mu-\beta)^2+2\theta\gamma^2}. 
\]

For the numerical illustration, we present the pdf's of the standard NIG distributions and the first passage time of the standard NIG process in Figure \ref{fig:stdNIG}. To draw figure, we use the fast Fourier transform method for the equation (\ref{eq:pdf_FFT}) with the ch.F of the first passage time. 

Let $X=(X(t))_{t\ge0}\sim \textup{stdNIG}(\theta,-\beta)$, and $Y=(Y(t))_{t\ge0}\sim \textup{stdNIG}(\theta,\beta)$ with parameter $\theta=1$ and $\beta = 1/2$, and let $B=(B(t))_{t\ge0}$ be the standard Brownian motion. We consider continuously approximated processes $X^c$ and $Y^c$ for $X$ and $Y$, respectively. 
In the both left and right plates, the solid curve is for $X$, the dashed curve is for $Y$ and the dash-dot curve is for $B$. The left plate, we show that the pdf's of $X(1)=X^c(1)$(solid) and $Y(1)=Y^c$(dashed), that are skewed left, and skewed right, respectively. The pdf of $B(1)$ (dash-dot) is a standard normal pdf, which is symmetric. 
Let $\tau_{X^c}(l)$, $\tau_{Y^c}(l)$ and $\tau_B(l)$ be the first passage time of $X^c$, $Y^c$, and $B$, respectively, for the level $l=3$. In the right plate, we presented the pdf's of $\tau_{X^c}(l)$ (solid), $\tau_{Y^c}(l)$(dashed) and $\tau_B(l)$(dash-dot). 
Since the distribution of $X(1)$ ($Y(1)$) is skewed left (right) while the distribution of $B(1)$ is symmetric, the mod of $\tau_{X^c}(l)$($\tau_{Y^c}(l)$) is located a little right (left) of the mod of $\tau_B(l)$.

By applying Lemma \ref{Lemma:ChfTau} to the process $X\sim \textup{NTS}(\alpha$, $\theta$, $\beta$, $\gamma$, $\mu)$, we find $\eta(u)$ satisfying (\ref{eq:martingalecondition2}) that is
\begin{align}\label{eq:martingaleconditionNTS}
0=iu +(\mu-\beta)\eta(u) -\frac{2\theta^{1-\frac{\alpha}{2}}}{\alpha}
\left(\left(\theta-\beta \eta(u)-\frac{\gamma^2 \eta(u)^2}{2}\right)^{\frac{\alpha}{2}}-\theta^{\frac{\alpha}{2}}\right),
\end{align}
which has no explicit solution, but we can find the solution numerically. Moreover, the NTS process is also a pure jump \levy~process, we find the ch.F of the first passage time using the continuously approximated process of the NTS process.

For the numerical illustration, we present the function $\eta(u)$, ch.F's, and pdf's of the standard NTS distributions and the pdf of the first passage time of the standard NTS process in Figure \ref{fig:firsthittingtime_NTS}. We consider two standard NTS processes $X\sim\textup{stdNTS}(\alpha,\theta,-\beta)$ and $Y\sim\textup{stdNTS}(\alpha,\theta,\beta)$ with parameters $\alpha=1.25$, $\theta = 1$ and $\beta = 0.3$, and the level $l=3$. The upper left and right plates are $\eta(u)$'s for $X$ and $Y$, respectively, which satisfy \eqref{eq:martingaleconditionNTS}. The middle left and right plates are the $e^{-l\eta(u)}$ for the stdNTS$(1.25, 1, -0.3)$ and stdNTS$(1.25, 1, 0.3)$, respectively. Pdf's of $X(1)$ and $Y(1)$ are the dashed and solid curves, respectively, on the bottom left plate. 

Let $X^c$ and $Y^c$ be continuously approximated process with respect to $X$ and $Y$. Pdf's of the first passage time of $X^c$ and $Y^c$ are numerically approximated as the dashed and solid curves, respectively, on the bottom right plate. Dash-dot curves of bottom left and right plates are pdf of standard normal distribution and pdf of the first passage time of the standard Brownian motion, respectively. For the same arguments as standard NIG case, we have mode of the dashed and solid curves are located left and right of the dash-dot curve, respectively in the bottom right plate.

Finally, the pdf obtained by the characteristic function is compared with the first passage time distribution of simulated sample paths. We denote the pdf's of $\tau_{X^c}(1)$ and $\tau_{Y^c}(1)$ as $f_{\tau_{X^c}}$ and $f_{\tau_{Y^c}}$, respectively.  We generate 20,000 sample paths of stdNTS$(1.25, 1, -0.3)$, using inverse transform method explained in \cite{RachevKimBianchiFabozzi:2011a}. That is  $\{X_j(n\varDelta t)| n=1,2,\cdots, 1,440\}$ with time step $\varDelta t = 1/48$ year fraction for $j=1,2,\cdots, 20,000$. Then we obtain 30 years sample path. Setting $l=3$ as above example, We find the set of the first hitting time as \[
\mathcal T = \{n_j\varDelta t | n_j=\min\{n \text{ such that } X_j(n\varDelta t)>l| n=1,\cdots, 1,440\}, j=1,\cdots, 20,000\}.
\] We present the relative histogram for $\mathcal T$ and $f_{\tau_{X^c}}$ on the left plate of Figure \ref{fig:simulation_stdNTS_HittingTime_distribution}. We do the same test for stdNTS$(1.25, 1, 0.3)$, and present the relative histogram for the simulation based first passage time of stdNTS$(1.25, 1, 0.3)$ and $f_{\tau_{Y^c}}$ on the right plate of Figure \ref{fig:simulation_stdNTS_HittingTime_distribution}. 


\subsection{Case of CGMY Process}
Let $\alpha\in(0,2)$, $C, \lambda_+, \lambda_->0$, and $\mu\in\R$. 
The pure jump \levy~process $X$ whose ch.F is equal to 
\begin{align*}
\phi_{X(t)}(u)  =\exp\Big(& (\mu-C\Gamma(1-\alpha)(\lambda_+^{\alpha-1}-\lambda_-^{\alpha-1})) iut \\
&-tC\Gamma(-\alpha)\left((\lambda_+-iu)^{\alpha}-\lambda_+^{\alpha}+(\lambda_-+iu)^{\alpha}-\lambda_-^{\alpha}\right)\Big)
\end{align*}
is referred to as CGMY tempered stable process (See \cite{CGMY:2002}) or CGMY process with parameters $(\alpha$, $C$, $\lambda_+$, $\lambda_-$, $\mu)$,and denoted by $X\sim \textup{CGMY}(\alpha$, $C$, $\lambda_+$, $\lambda_-$, $\mu)$. The CGMY process has finite exponential moments for a closed interval, That is, $E[e^{aX(t)}]<\infty$
if $a\in\left[-\lambda_-, \lambda_+\right]$.
If $X\sim \textup{CGMY}(\alpha$, $C$, $\lambda_+$, $\lambda_-$, $\mu)$ where $\mu = 0$ and $C = \left(\Gamma(2-\alpha)(\lambda_+^{\alpha-2}+\lambda_-^{\alpha-2})\right)^{-1}$ then $E[X(t)]=0$ and $\var(X(t))=t$ for all $t>0$.
In this case, the process $X$ is referred to as the \textit{standard CGMY process} with parameters $(\alpha$, $\lambda_+$, $\lambda_-)$ and denoted by $X\sim \textup{stdCGMY}(\alpha$, $\lambda_+$, $\lambda_-)$.

For given $u$, we find $\eta(u)$ satisfying (\ref{eq:martingalecondition2}) that is
\begin{align}\label{eq:martingaleconditionCTS}
0 =& iu +(\mu-C\Gamma(1-\alpha)(\lambda_+^{\alpha-1}-\lambda_-^{\alpha-1}))\eta(u) \\
\nonumber
&-C\Gamma(-\alpha)\left((\lambda_+-\eta(u))^{\alpha}-\lambda_+^{\alpha}+(\lambda_-+\eta(u))^{\alpha}-\lambda_-^{\alpha}\right)
\end{align}
which has no explicit solution. As the NTS process case in (\ref{eq:martingaleconditionNTS}), we are also able to find the solution of (\ref{eq:martingaleconditionCTS}) numerically, and ch.F is as the equation in Lemma \ref{Lemma:ChfTau}.

For the numerical illustration, we consider two standard CGMY processes $X\sim$ stdCGMY( $0.75$, $3$, $1$) and $Y\sim$stdCGMY($0.75$, $1$, $3$), and the level $l=3$. Also we take $X^c$ and $Y^c$ which are continuously approximated process of $X$ and $Y$, respectively.
We present the function $\eta(u)$, ch.F's, and pdf's of the standard CGMY distributions and the pdf of the first passage time of the standard CGMY process in Figure \ref{fig:firsthittingtime_CGMY}.  The upper left and right plates are the $\eta(u)$ for stdCGMY($0.75$, $3$, $1$) and stdCGMY($0.75$, $1$, $3$), respectively. The middle left and right plates are the chF's for stdCGMY($0.75$, $3$, $1$) and stdCGMY($0.75$, $1$, $3$), respectively. Pdf's of stdCGMY($0.75$, $3$, $1$) and stdCGMY($0.75$, $1$, $3$) are the dashed and solid curves, respectively, on the bottom left plate. Pdf's of the first passage times of $X^c$ and $Y^c$ are the dashed and solid curves, respectively, on the bottom right plate. Dash-dot curves of bottom left and right plates are pdf of standard normal distribution and pdf of the first passage time of the standard Brownian motion, respectively. For the same arguments as standard NIG case, we have mode of the dashed and solid curves are located left and right of the dash-dot curve, respectively in the bottom right plate.

We can do the same simulation experiment for the stdCGMY($0.75$, $1$, $3$) and stdCGMY( $0.75$, $3$, $1$) as the standard NTS case in the previous section. We omit to show the result since the result of CGMY simulation cases are very similar as of the NTS simulation cases.

\section{Application to Exotic Option Pricing and Numerical Illustration}
Let $X = (X(t))_{t\ge0}$ and $X^c = (X^c(t))_{t\ge0}$ be a \levy~process and its continuously approximated process, respectively, and suppose that there is an closed interval $I$ containing 0 such that $E[e^{a X(t)}]<\infty$ for any $a\in I$. We assume that there exist $\eta$ satisfying the condition of Lemma \ref{Lemma:ChfTau} for $X$, and $\tau(l)$ is the first hitting time given by (\ref{eq:tau(l)}). 

Let $r$ and $d$ be the risk free rate of return and the continuous dividend rate of a given underlying asset, respectively.
The underlying asset price process $(S(t))_{t\ge0}$ is assumed as
\[
S(t) = S(0) e^{X(t)}.
\]
All of the market models in this paper are based on the risk-neutral world which has no-arbitrage. So we assume that the discount price process $(\tilde S(t))_{t\ge0}$ with $\tilde S(t)=e^{-(r-d)t}S(t)$ is martingale. In this case, we referred to the risk-neutral price model as \textit{\levy~market model}.

The class of TS process is a subclass of \levy~process including NTS and CGMY process. The \levy~process option pricing models (i.e. \levy model) with TS processes are often used for the option pricing theory. The model can capture the volatility smile effect of option market by describing fattails and skewness of risk-neutral measure (See \cite{Boyarchenko_Levendorskii:2002}, \cite{Cont_Tankov:2004}, \cite{Schoutens:2003}, and \cite{RachevKimBianchiFabozzi:2011a}).

~\\
\noindent\textit{NTS Market model}\\
In the \levy~market model, suppose $(X(t))_{t\ge0}$ is the NTS process with parameters $(\alpha$, $\beta$, $\theta$, $\gamma$, $\mu)$, where
\[
\mu = r - d +\beta +\frac{2\theta^{1-\frac{\alpha}{2}}}{\alpha}
\left(\left(\theta-\beta -\frac{\gamma^2}{2}\right)^{\frac{\alpha}{2}}-\theta^{\frac{\alpha}{2}}\right),
\]
then the discount process $(\tilde S(t))_{t\ge0}$ is martingale, since
\[
E[\tilde S(t)]=e^{-(r-d)t}S(0)E[ \exp( X(t) )] =S(0).
\]
In this case, we say that the underlying asset price process follows the \textit{NTS market model}. In particular, if $\alpha=1$ we say that it follows \textit{NIG market model}.

~\\
\noindent\textit{CGMY Market model}\\
In the \levy~market model, suppose $(X(t))_{t\ge0}$ is the CGMY process with parameters $(\alpha, C, \lambda_+, \lambda_-, \mu)$, where
\[
\mu = r - d +C\Gamma(1-\alpha)(\lambda_+^{\alpha-1}-\lambda_-^{\alpha-1}) +C\Gamma(-\alpha)\left((\lambda_+-1)^{\alpha}-\lambda_+^{\alpha}+(\lambda_-+1)^{\alpha}-\lambda_-^{\alpha}\right),
\]
then the discount process $(\tilde S(t))_{t\ge0}$ is martingale for the same argument of the NTS market model case. In this case, we say that the underlying asset price process follows the \textit{CGMY market model}.

~\\
\noindent\textit{Model calibration}\\
We calibrate parameters of the three TS (NIG, NTS and CGMY) market models by using the S\&P 500 index call and put options. We obtain the S\&P 500 index call and put price from OptionMetrics$^{TM}$ provided by Wharton Research Data Services. The parameters are calibrated by the least square curve fit and the call and put option prices of NTS, NIG, and CGMY models are calculated by Fast Fourier Transform method by \cite{Lewis:2001} and \cite{CarrMadan:1999}. For the benchmark, we calibrate the BS model (\cite{BlackScholes:1973}) parameter $\sigma$ which is referred to as \textit{volatility}.

In this calibration, we use the option prices of October 8, 2014 on which underlying S\&P 500 index was $S(0)=1968.89$, risk-neutral interest rate was $r=0.12\%$, and dividend rate for S\&P 500 index was $d=1.94\%$. We select calls and puts with moneyness (= (strike price) / (current underlying price)) in between 0.9 and 1.1, and having time to maturities, 7 days, 32 days or 52 days. We calibrate parameters for call prices and put prices separately. The calibrated parameters are presented in Table \ref{Table:Calibration}. In the table, we provide three error estimators together:  the average absolute error (AAE), the average absolute error as a percentage of the mean price (APE), and the  root-mean-square error (RMSE),\footnote{See \cite{Schoutens:2003} for additional details.} defined as follows:
\begin{align*}
&    \textup{AAE} =
    \sum_{j=1}^N \frac{|{P}_j -
    \widehat{P}_j|}{N},
~~    \textup{APE} = \frac{\sum_{j=1}^N \frac{|{P}_j -
    \widehat{P}_i|}{N}}{\sum_{j=1}^N\frac{{P}_j}{N}},
~~    \textup{RMSE} = \sqrt{
    \sum_{j=1}^N \frac{({P}_j -
    \widehat{P}_j)^2}{N}},
\end{align*}
where $\widehat{P}_n$ and ${P}_n$ are model prices and observed market prices of options,  $n\in\{1,\ldots,N\}$, and $N$ is the number of observed call option prices.  As reported many literature including \cite{CGMY:2002}, \cite{Schoutens:2003}, and \cite{RachevKimBianchiFabozzi:2011a}, the NIG, NTS, and CGMY market models have better calibration performance than the BS model, that is, the three error estimators for those three \levy~market models are remarkably smaller than the three error estimator for the BS model.

~\\
\noindent\textit{Approximation for Characteristic function of $\tau(l)$}\\
As we discussed in Section \ref{section2}, the characteristic function of $\tau(l)$ is numerically approximated as $\phi_{\tau(l)}(u)\approx e^{-l\eta(u)}$ for NTS and CGMY market model. With this approximation we discuss pricing Perpetual American Call and Put, and pricing Barrier Option in the following subsections. In the following market, we assume that priving process $X$ of stock price process $S$ is the continuously approximated process for NTS (including NIG) and CGMY processes.

\subsection{Application to Perpetual American Call and Put}
We consider a perpetual American call and put options with strike price $K$. The perpetual call price is equal to 
\[
C_\text{perpetual}=\begin{cases}
\displaystyle \frac{K}{\eta^+(ir)-1}\left(\frac{S(0)(\eta^+(ir)-1)}{K\eta^+(ir)}\right)^{\eta^+(ir)} &\text{ if } S(0)\le L^+ \\
S(0)-K &\text{ if } S(0)>L^+
\end{cases},
\]
where
\[
L^+ = \frac{\eta^+(ir)K}{\eta^+(ir)-1},
\]
and $\eta^-(ir)$ is the value satisfying \eqref{eq:martingalecondition2} and \eqref{eq:chf_tau} for $l>0$ and $u=ir$.
The perpetual put price is equal to 
\[
P_\text{perpetual}=\begin{cases}
\displaystyle \frac{K}{1-\eta^-(ir)}\left(\frac{S(0)(\eta^-(ir)-1)}{K\eta^-(ir)}\right)^{\eta^-(ir)} &\text{ if } S(0)\ge L^- \\
K-S(0) &\text{ if } S(0)<L^-
\end{cases},
\]
where 
\[
L^- = \frac{\eta^-(ir)K}{\eta^-(ir)-1},
\]
and $\eta^-(ir)$ is the value satisfying \eqref{eq:martingalecondition2} and \eqref{eq:chf_tau} for $l<0$ and $u=ir$.
More details for the perpetual American call and put prices are presented in Appendix.

Using the perpetual call/put formula above, we calculate prices of the perpetual American call and put options for the NIG, NTS and CGMY market models using the calibrated parameters on October 8, 2014 presented in Table \ref{Table:Calibration}. We also use the underlying price, risk-neutral rate of return, and dividend rate from the data on October 8, 2014. Since perpetual call and put prices for NIG, NTS and CGMY market models are very similar, we present only prices (solid curves) under the CGMY market model in Figure \ref{fig:PerpetualCallPutPrices}. for the benchmark, we calculate perpetual American option prices (dash-dot curves) using BS model. In this context, the left plate of the figure is perpetual call prices and tie right plate is the perpetual put prices for moneyness from 0.8 to 1.2. We find that BS prices are more or less smaller than CGMY prices.

\subsection{Application to Barrier Option Pricing}
The barrier option is one of the most popular exotic option. 
Barrier options are classified by the knock-in barrier option, and the knock-out barrier option. The knock-in barrier option is activated when the underlying asset price hit a given barrier level. The knock-out barrier option is alive until the underlying price hit the barrier level, but once the underlying price hit the barrier level, the price of the option becomes zero. The knock-in barrier option is divided by the up-and-in barrier option and the down-and-in barrier option. If the barrier level is upper (lower) than the current underlying prices, then the knock-in barrier option is referred to as the up-and-in (down-and-in) barrier option. The knock-out barrier option is also divided by the up-and-out barrier option and the down-and-out barrier option. If the barrier level is upper (lower) than the current underlying prices, then the knock-out barrier option is referred to as the up-and-out (down-and-out) barrier option\footnote{See \cite{Hull:2015} for more details.}. 

In this section, we discuss pricing of European style barrier call and put options.
Let $B$ be the barrier level, $K$ be the strike price of call and put, and $T$ be the time to maturity. Suppose that the current underlying asset price is $S(0)$, and let $l = \log(B/S(0))$. Then we have the down-and-in and down-and-out options if $l<0$ and the up-and-in and up-and-out options if $l>0$. Let $r$ and $d$ be the risk free rate of return and the continuous dividend rate of a given underlying asset, respectively.

The down-and-in call ($c_{di}$), up-and-in put ($p_{ui}$) are priced by 
\[
c_{di} = \frac{e^{-rT}K^{1+\rho}}{(2\pi)^2 B^\rho}\int_{-\infty}^\infty \left(\frac{B}{K}\right)^{iu} \left(\frac{H(u)}{(\rho-iu)(1+\rho-iu)}\right)du
, ~~~\rho<-1
\]
and
\[
p_{ui} = \frac{e^{-rT}K^{1+\rho}}{(2\pi)^2 B^\rho}\int_{-\infty}^\infty \left(\frac{B}{K}\right)^{iu} \left(\frac{H(u)}{(\rho-iu)(1+\rho-iu)}\right)du
, ~~~\rho>0,
\]
and the up-and-in call ($c_{ui}$) and down-and-in put ($p_{di}$) are priced by 
\[
c_{ui} = \begin{cases}
\frac{e^{-rT}K^{1+\rho}}{(2\pi)^2 B^\rho}\int_{-\infty}^\infty \left(\frac{B}{K}\right)^{iu} \left(\frac{H(u)}{(\rho-iu)(1+\rho-iu)}\right)du & \text{ if } K\le B \\
\frac{e^{-rT}K^{1+\rho}}{2\pi S(0)^\rho}\int_{-\infty}^\infty \left(\frac{S(0)}{K}\right)^{iu} \left(\frac{\phi_{X(T-t)}(u+i\rho)}{(\rho-iu)(1+\rho-iu)}\right)du & \text{ if } K> B
\end{cases}
, ~~~\rho<-1,
\]
and
\[
p_{di} = \begin{cases}
\frac{e^{-rT}K^{1+\rho}}{(2\pi)^2 B^\rho}\int_{-\infty}^\infty \left(\frac{B}{K}\right)^{iu} \left(\frac{H(u)}{(\rho-iu)(1+\rho-iu)}\right)du & \text{ if } K\ge B \\
\frac{e^{-rT}K^{1+\rho}}{2\pi S(0)^\rho}\int_{-\infty}^\infty \left(\frac{S(0)}{K}\right)^{iu} \left(\frac{\phi_{X(T-t)}(u+i\rho)}{(\rho-iu)(1+\rho-iu)}\right)du & \text{ if } K< B
\end{cases}
, ~~~\rho>0,
\]
where
\[
H(u) = \int_{-\infty}^\infty \frac{e^{T\psi_X(u+i\rho)}-e^{-ivT}}{\psi_X(u+i\rho)+iv}\phi_{\tau(l)}(v)dv.
\]
Since knock-out call(put) option price can be calculated by the vanilla call(put) price minus knock-in call(put) price, we do not consider the knock-out call and put options in this section. More Details and proofs for the barrier option pricing are explained in Appendix. 

Notably, we calculate those four knock-in option prices numerically. We use the estimated parameters in Table \ref{Table:Calibration} for call and put prices of October 8, 2014, on which underlying S\&P 500 index was $S(0)=1968.89$, risk-neutral interest rate was $r=0.12\%$, and dividend rate for S\&P 500 index was $d=1.94\%$. We assume that the time to maturity is $T$ = 1 year and we select strike prices between 1600 and 2300, that is, $K\in\{1600, 1625, \cdots,  2300\}$. Since the barrier option prices under NIG and NTS market models are not remarkably different from the prices under the CGMY market models, we just show the barrier option prices under CGMY market model. For the benchmark, we compare the CGMY prices to the barrier option prices based on BS model. 

In Figure \ref{fig:DownAndInCallPutPrices}, the down-and-in call and put option prices are presented in the left and right plates, respectively, for the barrier level $B$ is 1750. 
The barrier option on the CGMY model is typically greater than the price on BS model, since the CGMY distribution is skewed left and the CGMY distribution has fatter tails than Gaussian distribution. because of the left skewness of the CGMY distribution, the probability that the knock-in option be alive is larger than the case of BS model.
In Figure \ref{fig:UpAndInCallPutPrices}, the up-and-in call and put option prices are presented in the left and right plates, respectively, for the barrier level $B$ is 2200. That being the case, the BS barrier option prices are different from CGMY barrier option prices. 

In addition, we implement the four barrier option pricing formula using the fast Fourier transform algorithm. To accomplish this numerical calculation, we use Matlab$^{TM}$ 2015a on a PC equipped with Intel CORE i7$^{TM}$ processor (3.00GHz Dual core) and MS Windows$^{TM}$ 10 operating system. 
It took 59.9 seconds to obtain the down-and-in call option prices, and took 67.6 seconds to obtain down-and-in put option prices.
It took 69.3 seconds to obtain down-and-in call option prices, and took 61.4 seconds to obtain down-and-in put option prices.

\section{Conclusion}
In this paper, we found an approximation of the characteristic function of the first passage time for a \levy~process using the martingale approach and the continuously approximated process. More precisely, we found approximations of ch.F's and pdf's of the first passage time of three TS processes (NIG, NTS, and CGMY processes), explicitly or numerically. It was applied to price the perpetual American option and the barrier option. Numerical illustrations were provided together, for the calibrated parameters using the market call and put prices. We obtained one tractable method to find those exotic options under the TS market model. The result can be used for analyzing default probability in credit risk management also.

~\\
\textbf{Acknowledgment}
I am grateful to Professor Kyuong Jin Choi, in Haskayne School of Business, University of Calgary, who gave the motivation to complete of this research. Also, all remaining errors are entirely my own.

~\\
~\\

\bibliographystyle{plainnat}
\bibliography{refs_aaron_TS_perpetual_american}

\begin{thebibliography}{22}
\providecommand{\natexlab}[1]{#1}
\providecommand{\url}[1]{\texttt{#1}}
\expandafter\ifx\csname urlstyle\endcsname\relax
  \providecommand{\doi}[1]{doi: #1}\else
  \providecommand{\doi}{doi: \begingroup \urlstyle{rm}\Url}\fi

\bibitem[Applebaum(2004)]{Applebaum:2004}
D.~Applebaum.
\newblock \emph{\levy~process and stochastic calculus}.
\newblock Cambridge Univ. Press, New York, 2004.

\bibitem[Barndorff-Nielsen(1977)]{Barndorff-Nielsen401}
O.~Barndorff-Nielsen.
\newblock Exponentially decreasing distributions for the logarithm of particle
  size.
\newblock \emph{Proceedings of the Royal Society of London A: Mathematical,
  Physical and Engineering Sciences}, 353\penalty0 (1674):\penalty0 401--419,
  1977.
\newblock ISSN 0080-4630.
\newblock \doi{10.1098/rspa.1977.0041}.
\newblock URL
  \url{http://rspa.royalsocietypublishing.org/content/353/1674/401}.

\bibitem[Barndorff-Nielsen and
  Levendorskii(2001)]{BarndorffNielsenLevendorskii:2001}
O.~E. Barndorff-Nielsen and S.~Levendorskii.
\newblock Feller processes of normal inverse gaussian type.
\newblock \emph{Quantitative Finance}, 1:\penalty0 318 -- 331, 2001.

\bibitem[Barndorff-Nielsen and Shephard(2001)]{BarndorffNielsenShephard:2001}
O.~E. Barndorff-Nielsen and N.~Shephard.
\newblock Normal modified stable processes.
\newblock \emph{Economics Series Working Papers from University of Oxford,
  Department of Economics}, 72, 2001.

\bibitem[Black and Scholes(1973)]{BlackScholes:1973}
F.~Black and M.~Scholes.
\newblock The pricing of options and corporate liabilities.
\newblock \emph{The Journal of Political Economy}, 81\penalty0 (3):\penalty0
  637--654, 1973.

\bibitem[{Boguslavskaya}(2014)]{2014arXiv1403.1816B}
E.~{Boguslavskaya}.
\newblock {Solving optimal stopping problems for \levy~processes in infinite
  horizon via $A$-transform}.
\newblock \emph{ArXiv e-prints}, March 2014.

\bibitem[Boyarchenko and Levendorski\u{i}(2000)]{Boyarchenko_Levendorskii:2000}
S.~I. Boyarchenko and S.~Z. Levendorski\u{i}.
\newblock Option pricing for truncated {L\'evy} processes.
\newblock \emph{International Journal of Theoretical and Applied Finance},
  3:\penalty0 549--552, 2000.

\bibitem[Boyarchenko and
  Levendorski\u{i}(2002{\natexlab{a}})]{Boyarchenko_Levendorskii:2002}
S.~I. Boyarchenko and S.~Z. Levendorski\u{i}.
\newblock \emph{Non-Gaussian Merton-Black-Scholes Theory}.
\newblock World Scientific, 2002{\natexlab{a}}.

\bibitem[Boyarchenko and
  Levendorski\u{i}(2002{\natexlab{b}})]{Boyarchenko_Levendorskii:2002a}
S.~I. Boyarchenko and S.~Z. Levendorski\u{i}.
\newblock Perpetual american options under \levy~processes.
\newblock \emph{SIAM Journal on Control and Optimization}, 40\penalty0
  (6):\penalty0 1663--1696, 2002{\natexlab{b}}.

\bibitem[Boyarchenko and
  Levendorski\u{i}(2002{\natexlab{c}})]{Boyarchenko_Levendorskii:2002b}
S.~I. Boyarchenko and S.~Z. Levendorski\u{i}.
\newblock Barrier options and touch-and-out options under regular lévy
  processes of exponential type.
\newblock \emph{The Annals of Applied Probability}, 12\penalty0 (4):\penalty0
  1261--1298, 2002{\natexlab{c}}.

\bibitem[Carr and Madan(1999)]{CarrMadan:1999}
P.~Carr and D.~Madan.
\newblock Option valuation using the fast fourier transform.
\newblock \emph{Journal of Computational Finance}, 2\penalty0 (4):\penalty0
  61--73, 1999.

\bibitem[Carr et~al.(2002)Carr, Geman, Madan, and Yor]{CGMY:2002}
P.~Carr, H.~Geman, D.~Madan, and M.~Yor.
\newblock The fine structure of asset returns: An empirical investigation.
\newblock \emph{Journal of Business}, 75\penalty0 (2):\penalty0 305--332, 2002.

\bibitem[Chandra and Mukherjee(2016)]{math4010002}
S.~R. Chandra and D~Mukherjee.
\newblock Barrier option under lévy model : A pide and mellin transform
  approach.
\newblock \emph{Mathematics}, 4\penalty0 (1), 2016.

\bibitem[Cont and Tankov(2004)]{Cont_Tankov:2004}
R.~Cont and P.~Tankov.
\newblock \emph{Financial Modelling with Jump Processes}.
\newblock Chapman \& Hall / CRC, 2004.

\bibitem[Gerber and Shiu(1994)]{GerberShiu:1994}
H.~U. Gerber and E.~S.~W. Shiu.
\newblock Martingale approach to pricing perpetual american options.
\newblock \emph{Astin Bulletin}, 24:\penalty0 195--220, 1994.

\bibitem[Hull(2015)]{Hull:2015}
J.~C. Hull.
\newblock \emph{Options, Futures and Other Derivatives}.
\newblock Prentice-Hall, 9th edition, 2015.

\bibitem[Hurd and Kuznetsov(2009)]{hurd2009}
T.~R. Hurd and A.~Kuznetsov.
\newblock On the first passage time for brownian motion subordinated by a
  \levy~process.
\newblock \emph{Journal of Applied Probability}, 46\penalty0 (1):\penalty0
  181--198, 03 2009.
\newblock \doi{10.1239/jap/1238592124}.

\bibitem[Koponen(1995)]{Koponen:1995}
I.~Koponen.
\newblock Analytic approach to the problem of convergence of truncated {L\'evy}
  flights towards the gaussian stochastic process.
\newblock \emph{Physical Review E}, 52:\penalty0 1197--1199, 1995.

\bibitem[Lewis(2001)]{Lewis:2001}
A.~L. Lewis.
\newblock A simple option formula for general jump-diffusion and other
  exponential {L\'evy} processes.
\newblock \emph{avaible from http://www.optioncity.net}, 2001.

\bibitem[Rachev et~al.(2011)Rachev, Kim, Bianch, and
  Fabozzi]{RachevKimBianchiFabozzi:2011a}
S.~T. Rachev, Y.~S. Kim, M.~L. Bianch, and F.~J. Fabozzi.
\newblock \emph{Financial Models with {L\'{e}vy} Processes and Volatility
  Clustering}.
\newblock John Wiley \& Sons, 2011.

\bibitem[Rogers(2000)]{rogers2000}
L.~C.~G. Rogers.
\newblock Evaluating first-passage probabilities for spectrally one-sided lévy
  processes.
\newblock \emph{Journal of Applied Probability}, 37\penalty0 (4):\penalty0
  1173--1180, 12 2000.
\newblock \doi{10.1239/jap/1014843099}.

\bibitem[Schoutens(2003)]{Schoutens:2003}
W.~Schoutens.
\newblock \emph{{L\'evy} Processes in Finance}.
\newblock Wiley, 2003.

\end{thebibliography}

\clearpage

\begin{table}
\begin{center}
\caption{\label{Table:Calibration} Model Calibration}
\begin{tabular}{ccccccc}
\hline
Call/Put & Model & Parameters & AAE & APE & RMSE \\
\hline
Call  
 & BS & $\sigma=0.1267$ & 3.3619 & 0.1042 & 4.1806 \\ 
 \cline{2-6}
 & NIG & $\theta=5.1045$ , $\beta=-0.3356$ , & 2.4056 & 0.0746 & 2.8746 \\ 
 &      & $\gamma=0.1042$ \\
 \cline{2-6}
 & NTS & $\alpha=1.0808$ , $\theta=5.2150$ ,  & 2.3939 & 0.0742  & 2.8637 \\ 
 &      & $\beta=-0.3647$ , $\gamma=0.1010$ \\
 \cline{2-6}
 & CGMY & $\alpha=0.7250$ , $C=0.5019$ , & 2.4107 & 0.0747 & 2.8810 \\ 
 &      & $ \lambda_+=73.5549$ , $\lambda_-=11.5265$ \\
\hline 
Put  
 & BS & $\sigma=0.1396$ & 3.5132 & 0.1507 & 4.3561 \\ 
 \cline{2-6}
 & NIG & $\theta=4.2571$ , $\beta=-0.3466$ , & 1.3995 & 0.0600 &  1.6910 \\ 
 &     &  $\gamma=0.0984$\\
 \cline{2-6}
 & NTS & $\alpha=1.0980$ , $\theta=4.4348$ ,  & 1.3920 & 0.0597 & 1.6846 \\ 
 &      &  $\beta=-0.3864$ , $\gamma=0.0930$\\
 \cline{2-6}
 & CGMY & $\alpha=0.7066$ , $C=0.4743$ ,  & 1.4006 & 0.0601 & 1.6957 \\ 
 &      &  $\lambda_+=81.2931$ , $\lambda_-=9.7529$ \\
 \hline
 \end{tabular}
\end{center}
\end{table}

\begin{figure}[t]
\hspace{-1cm}
\includegraphics[width=7cm]{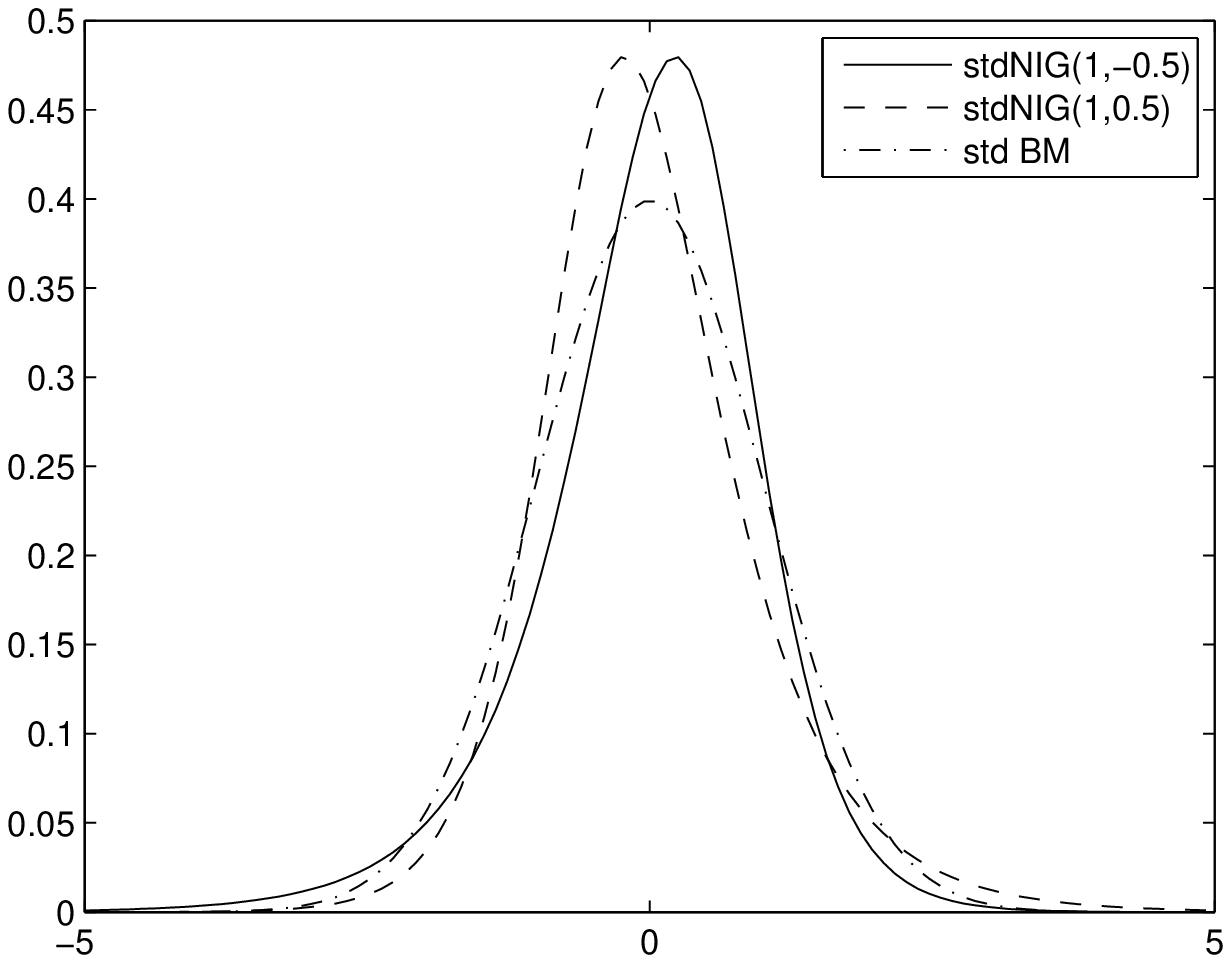}
\includegraphics[width=7cm]{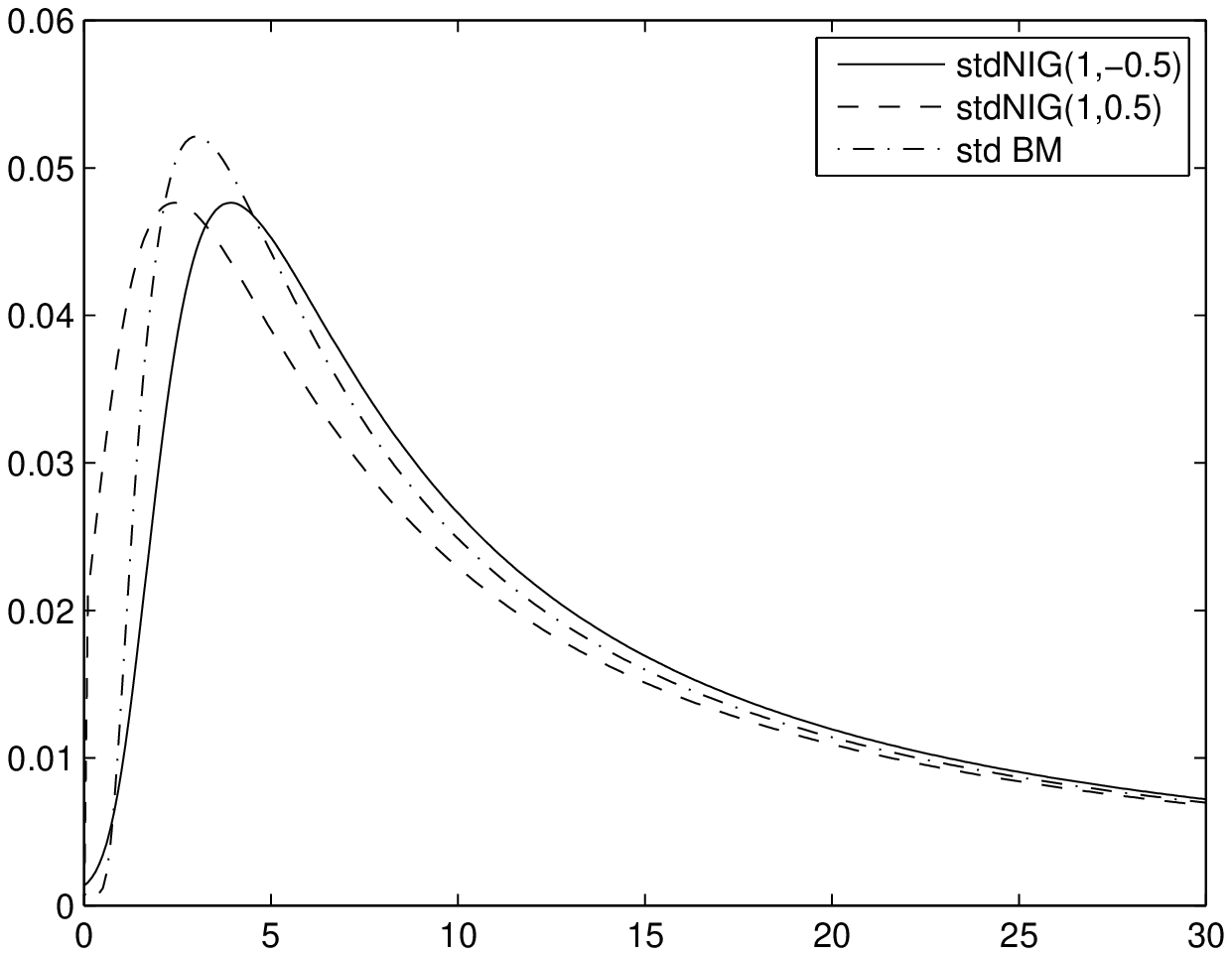}
\caption{\label{fig:stdNIG}Pdf's of standard NIG distributions (left) and the first passage time of standard NIG processes (right). }
\end{figure}
\clearpage

\begin{figure}[t]
\begin{center}
\hspace{-1cm}
\includegraphics[width=6.5cm]{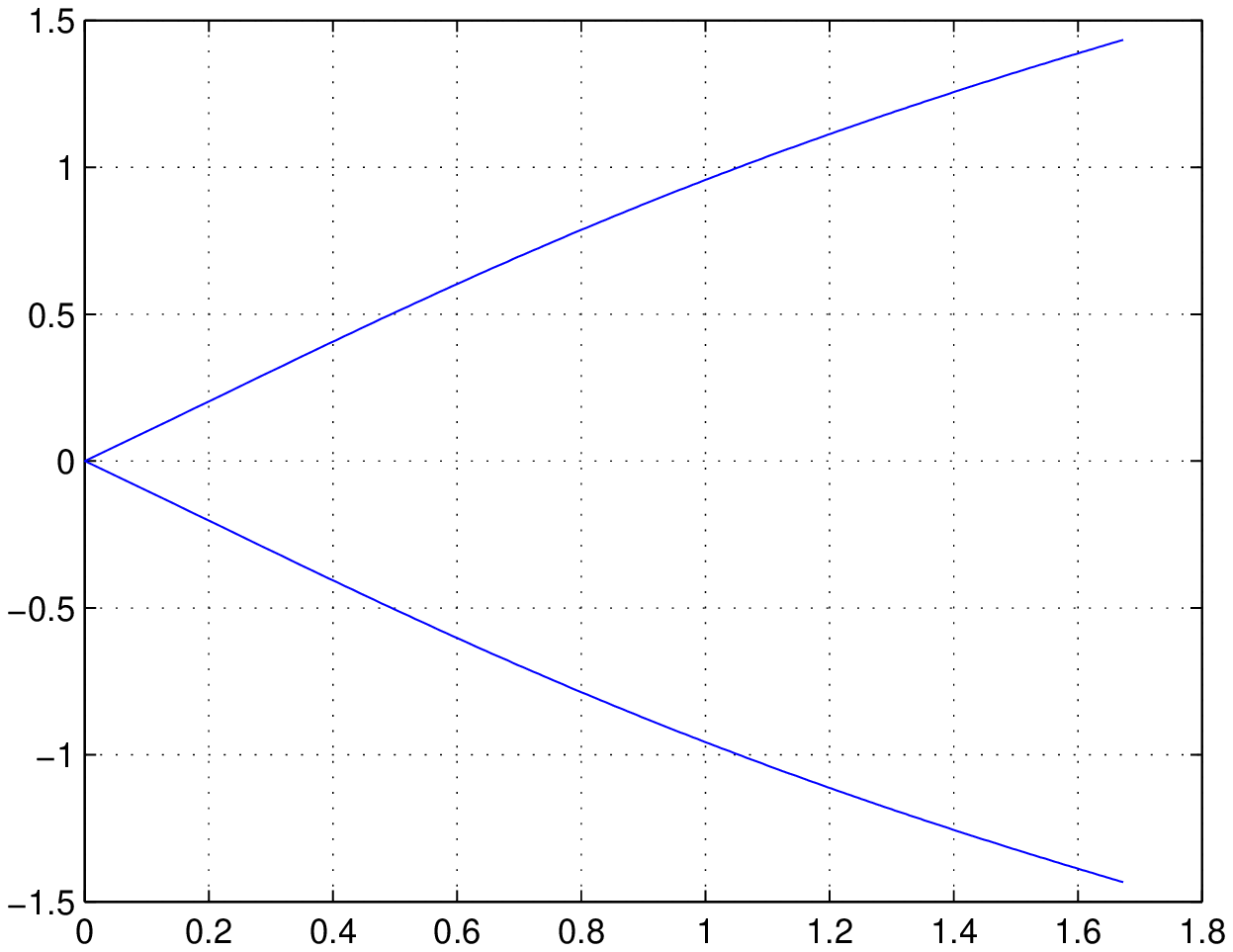}
\includegraphics[width=6.5cm]{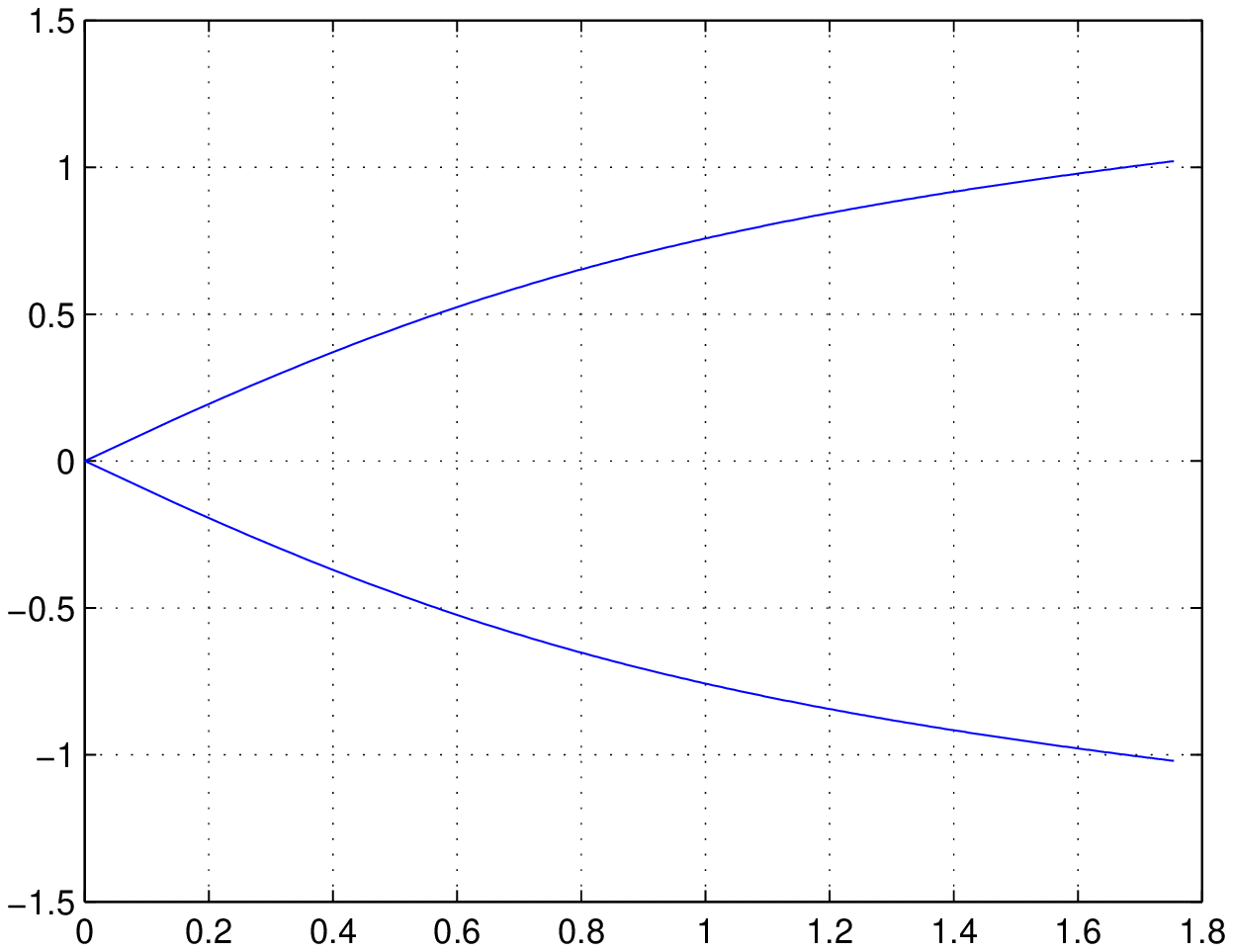}\\
\hspace{-1cm}
\includegraphics[width=6.5cm]{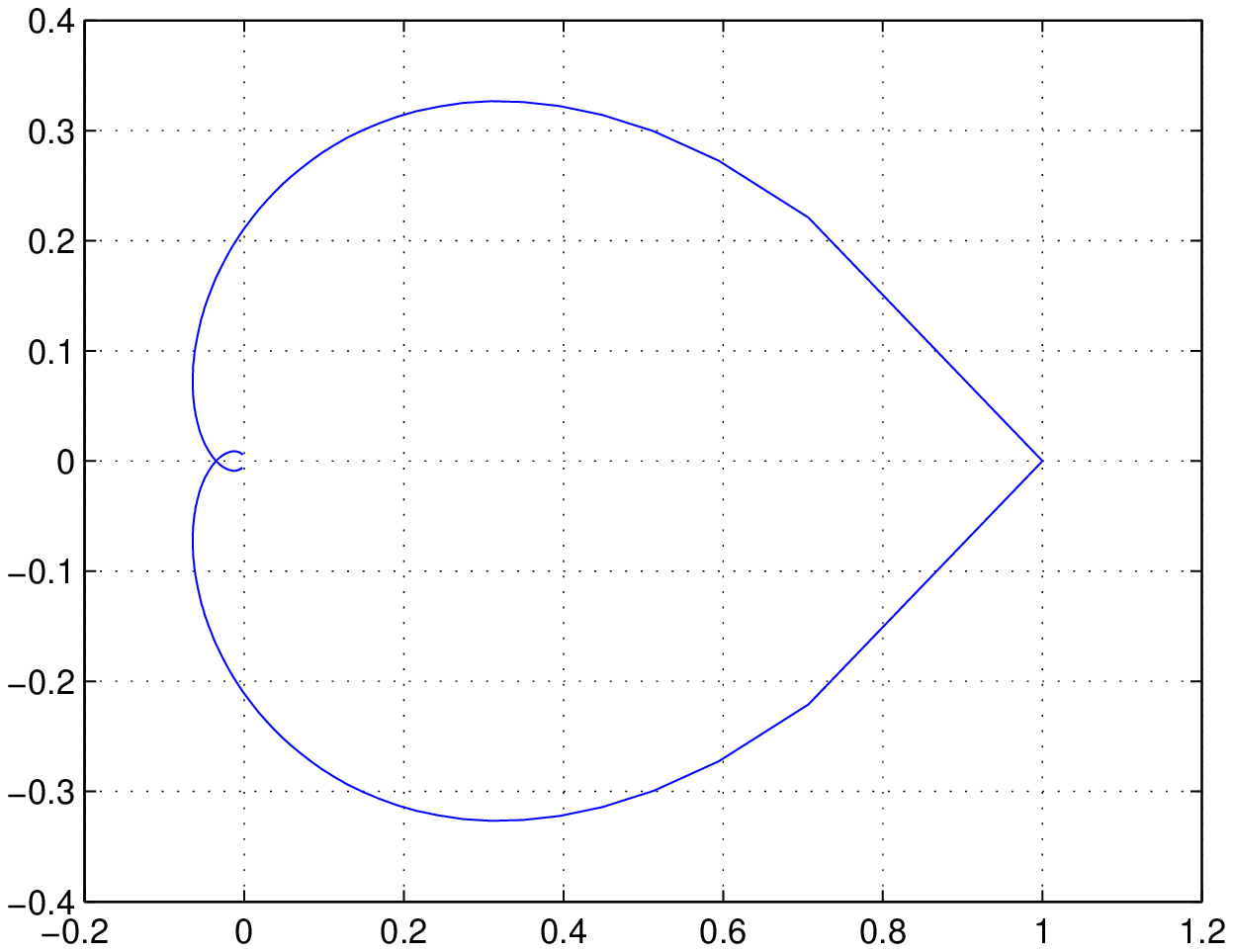}
\includegraphics[width=6.5cm]{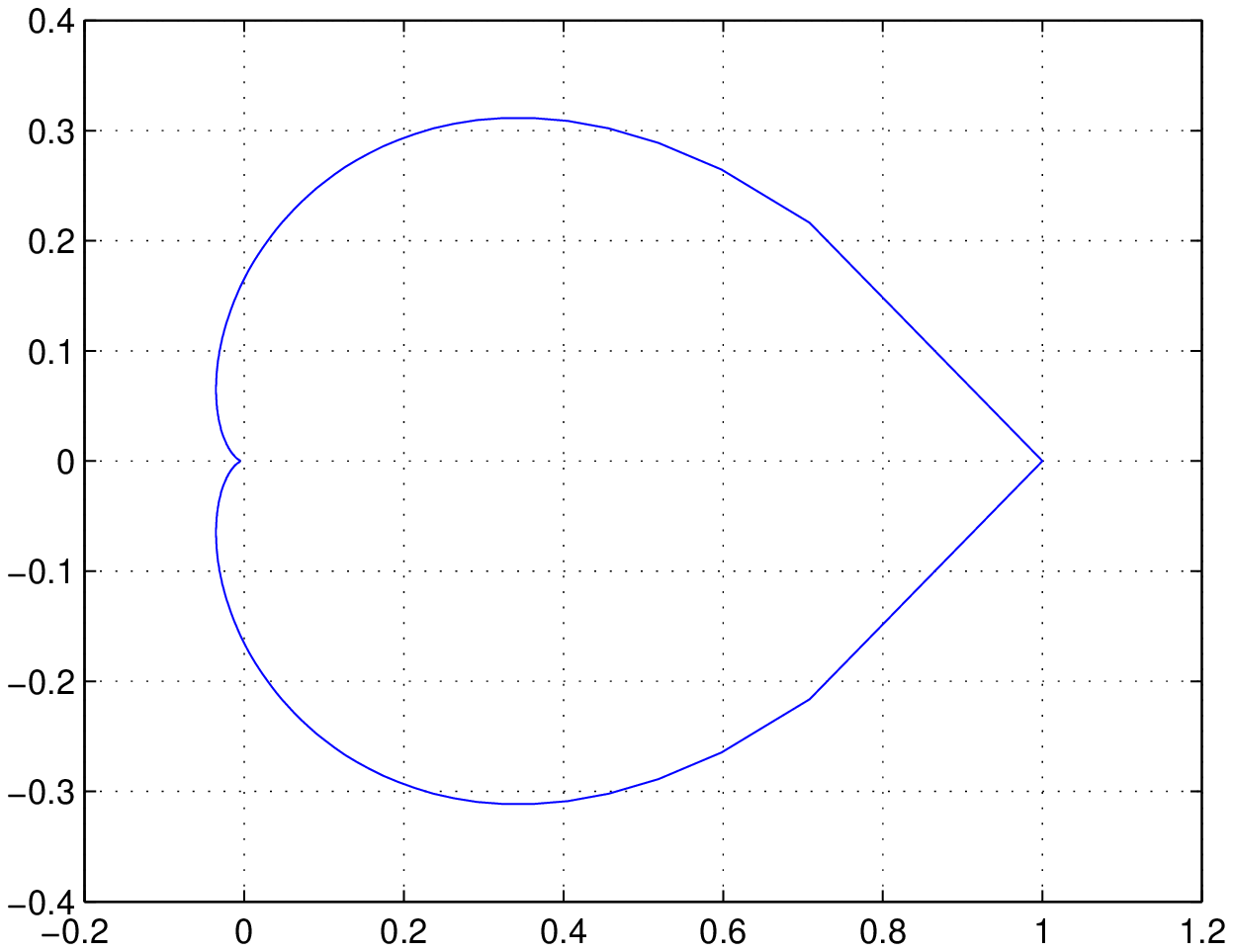}\\
\hspace{-1cm}
\includegraphics[width=6.5cm]{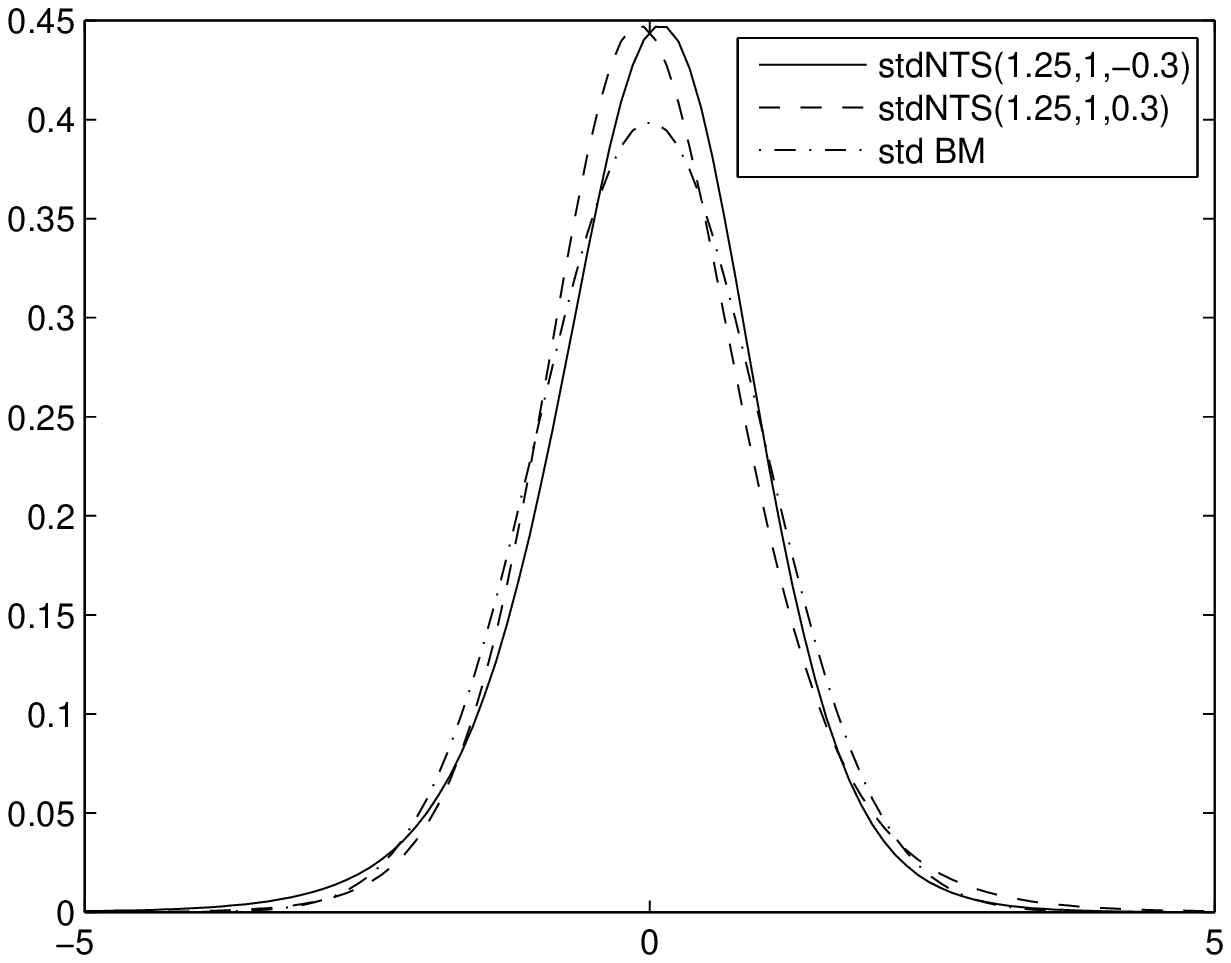}
\includegraphics[width=6.5cm]{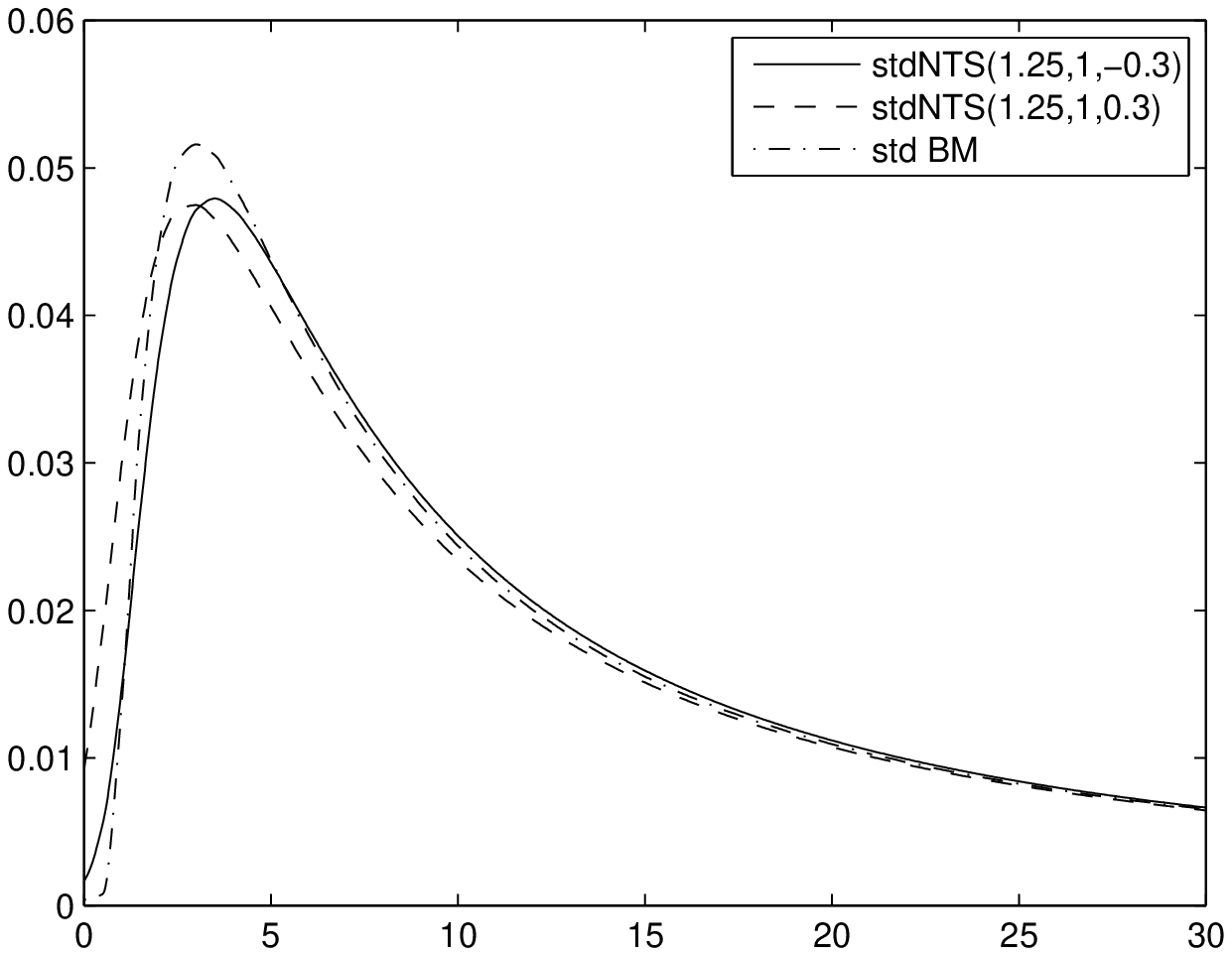}
\end{center}
\caption{\label{fig:firsthittingtime_NTS}Function $\eta(u)$'s, ch.F's and pdf's for first passage time of standard NTS processes for the level $l=3$.  The upper left is the function $\eta(u)$ for stdNTS$(1.25, 1, 0.3)$ and the upper right is for stdNTS$(1.25, 1, -0.3)$. The middle left is the characteristic function for the stdNTS$(1.25, 1, 0.3)$ and the middle right is for stdNTS$(1.25, 1, -0.3)$. Pdf's of standard NTS distributions are on the bottom left. The pdfs of the first passage time of standard NTS processes are on the bottom right.  }
\end{figure}
\clearpage

\begin{figure}[t]
\begin{center}
\hspace{-1cm}
\includegraphics[width=6.5cm]{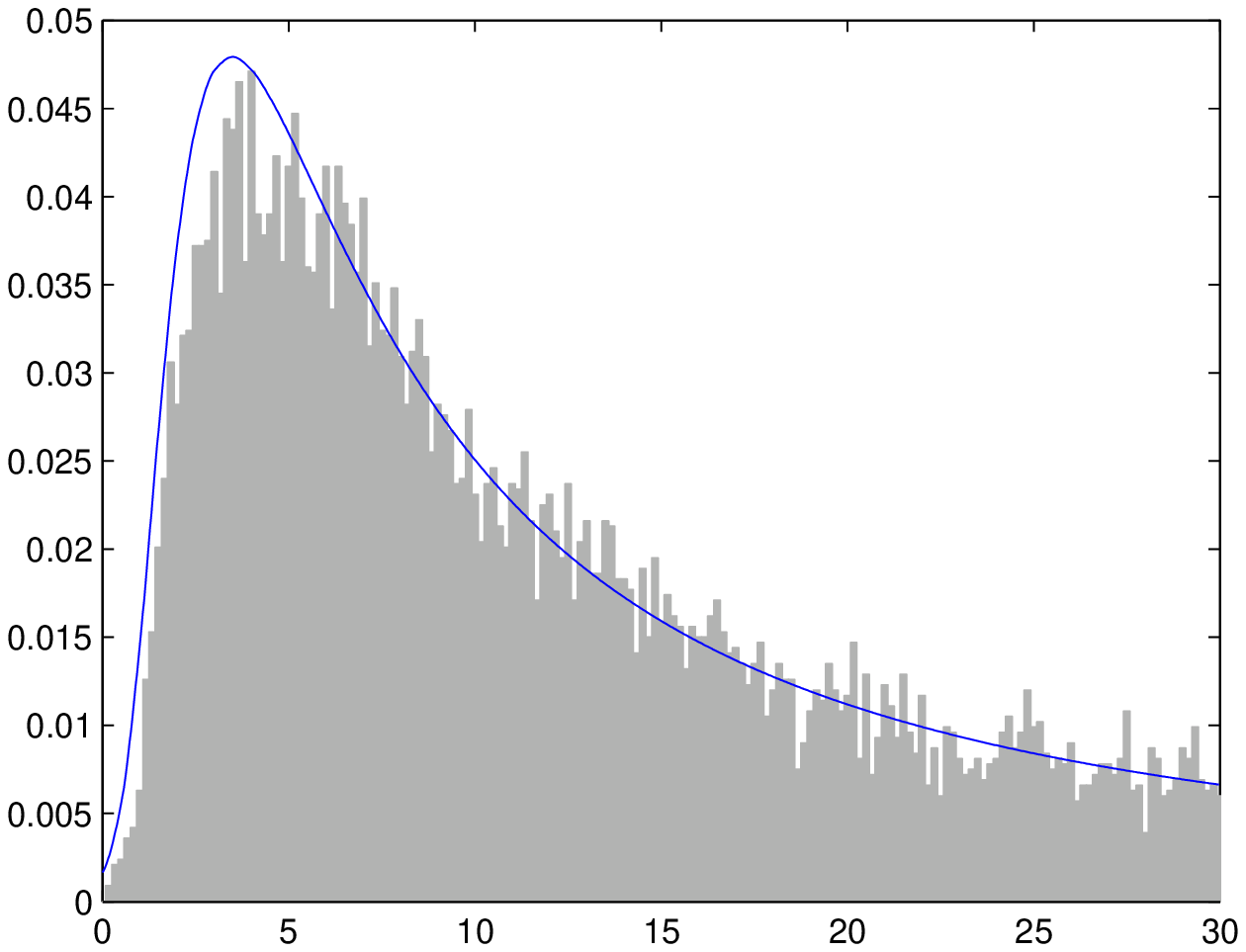}
\includegraphics[width=6.5cm]{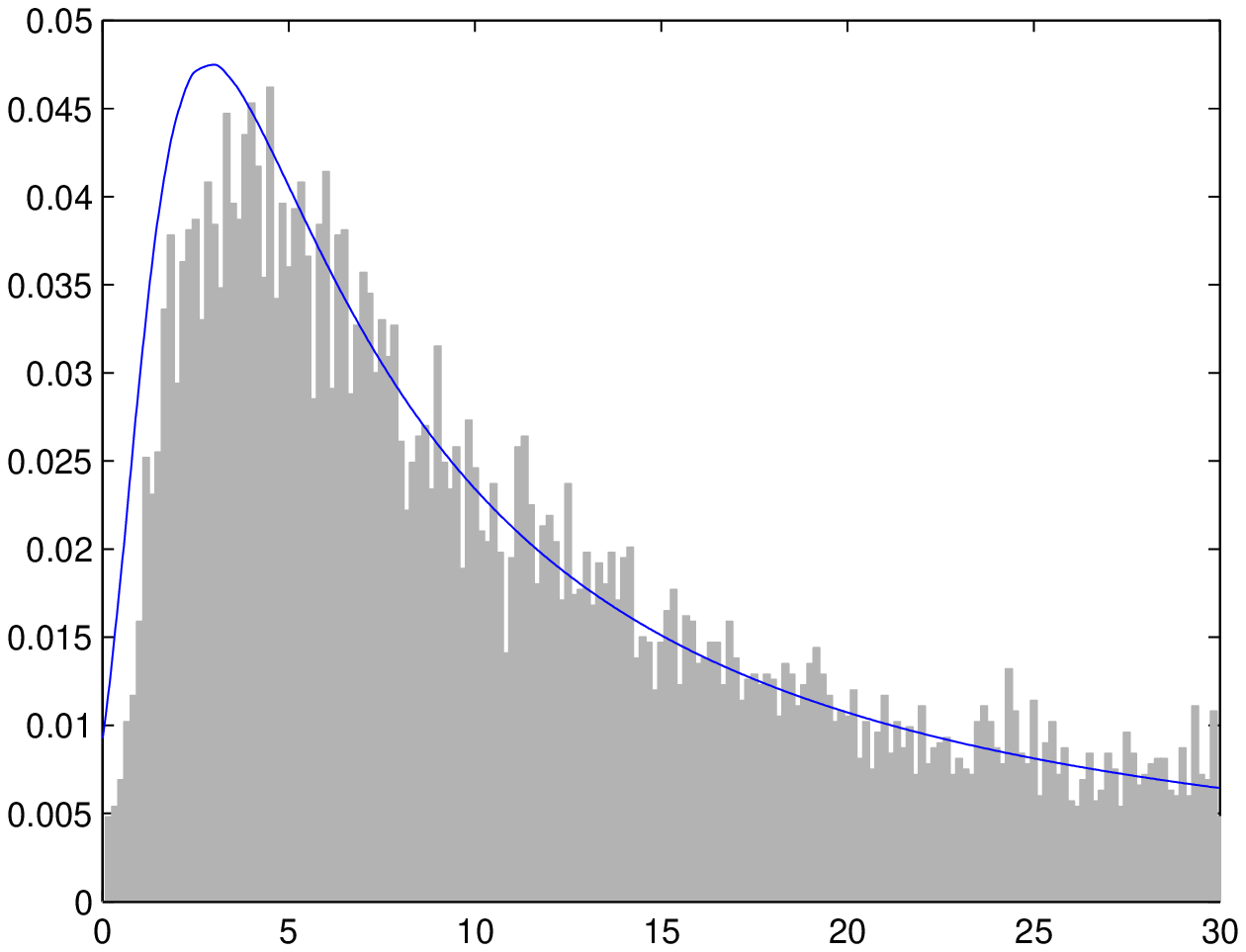}
\end{center}
\caption{\label{fig:simulation_stdNTS_HittingTime_distribution}
Relative histograms for first passage times for stdNTS$(1.25, 1, -0.3)$ simulated sample path (left) and stdNTS$(1.25, 1, 0.3)$ simulated sample path (right). }
\end{figure}
\clearpage

\begin{figure}[t]
\begin{center}
\hspace{-1cm}
\includegraphics[width=6.5cm]{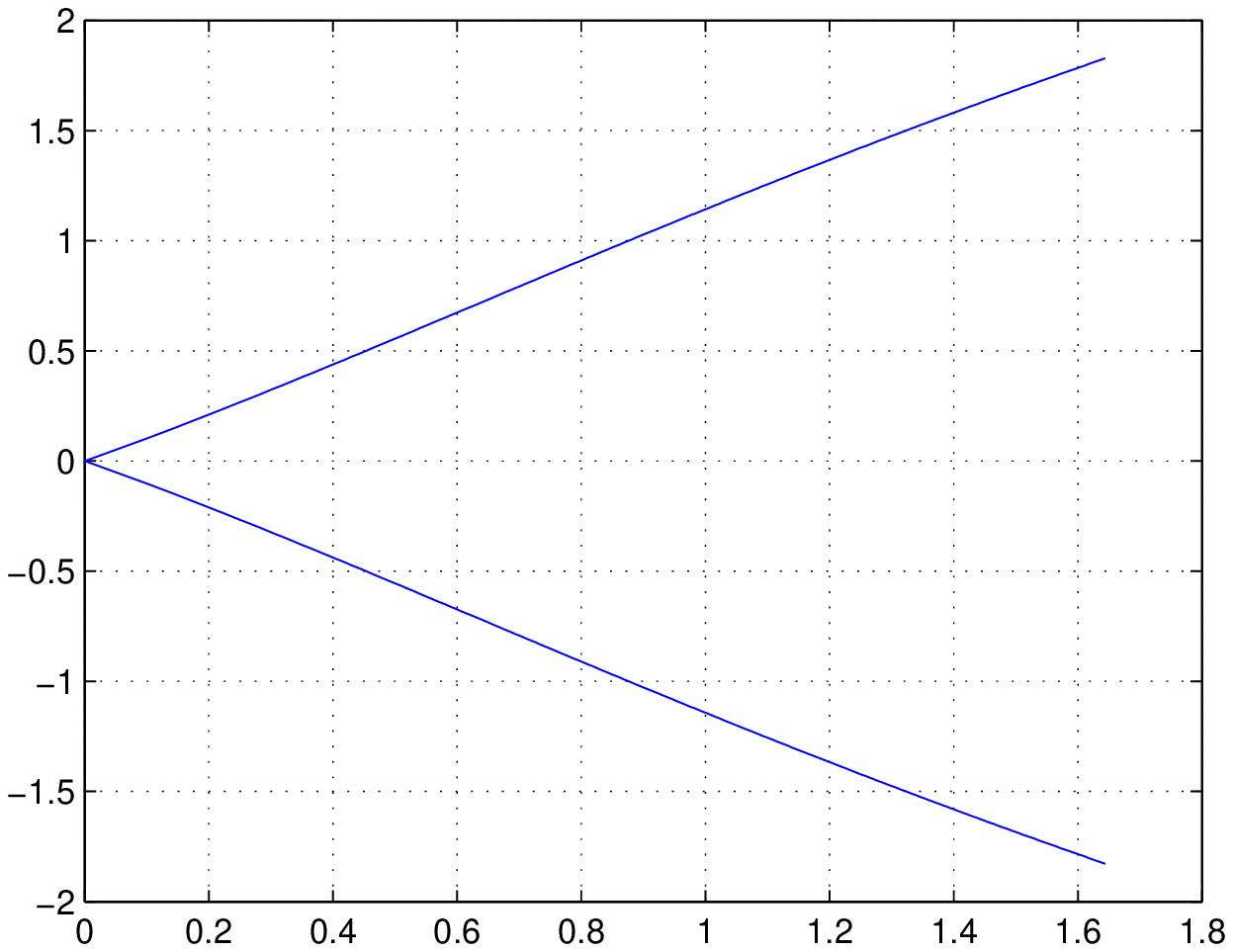}
\includegraphics[width=6.5cm]{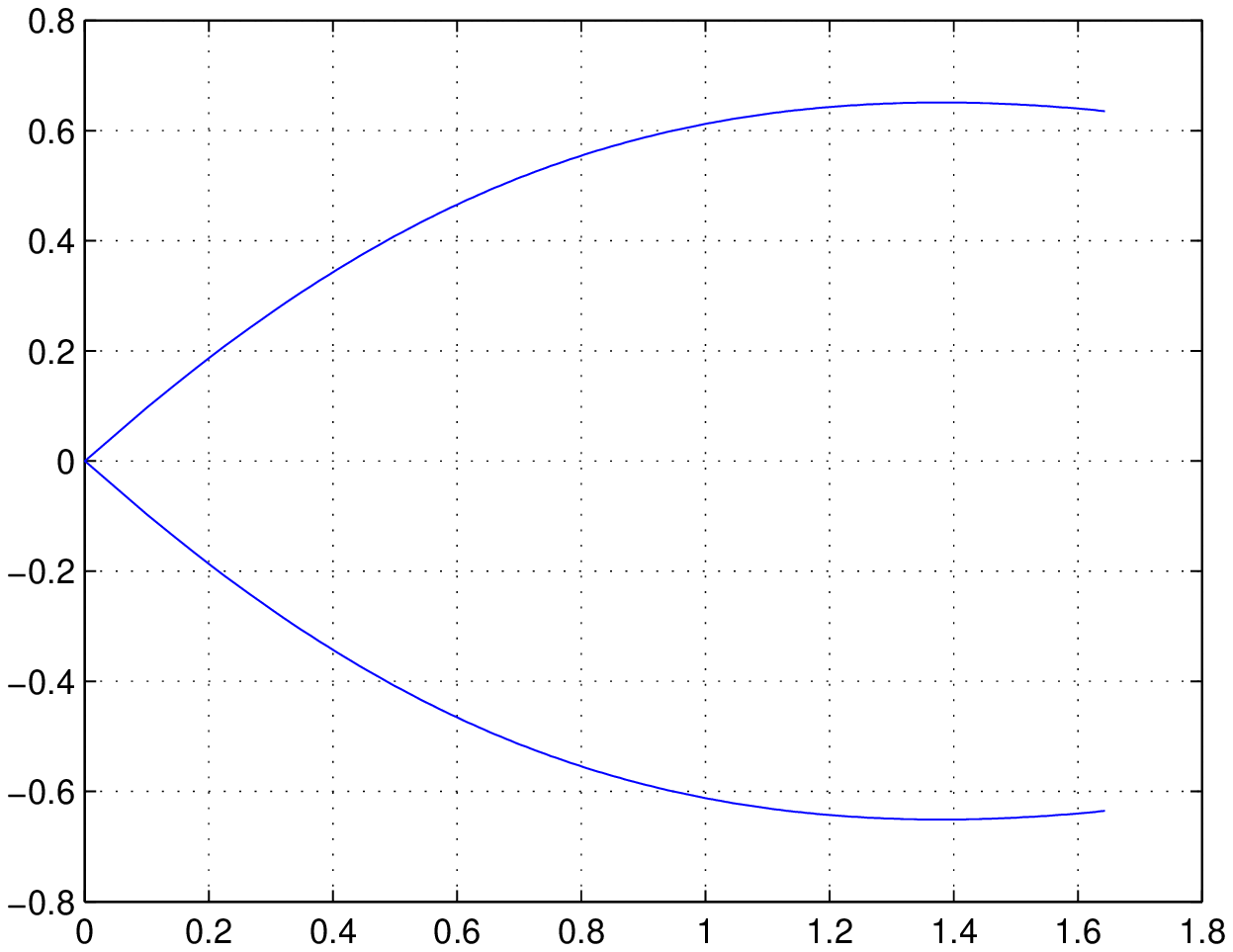}\\
\hspace{-1cm}
\includegraphics[width=6.5cm]{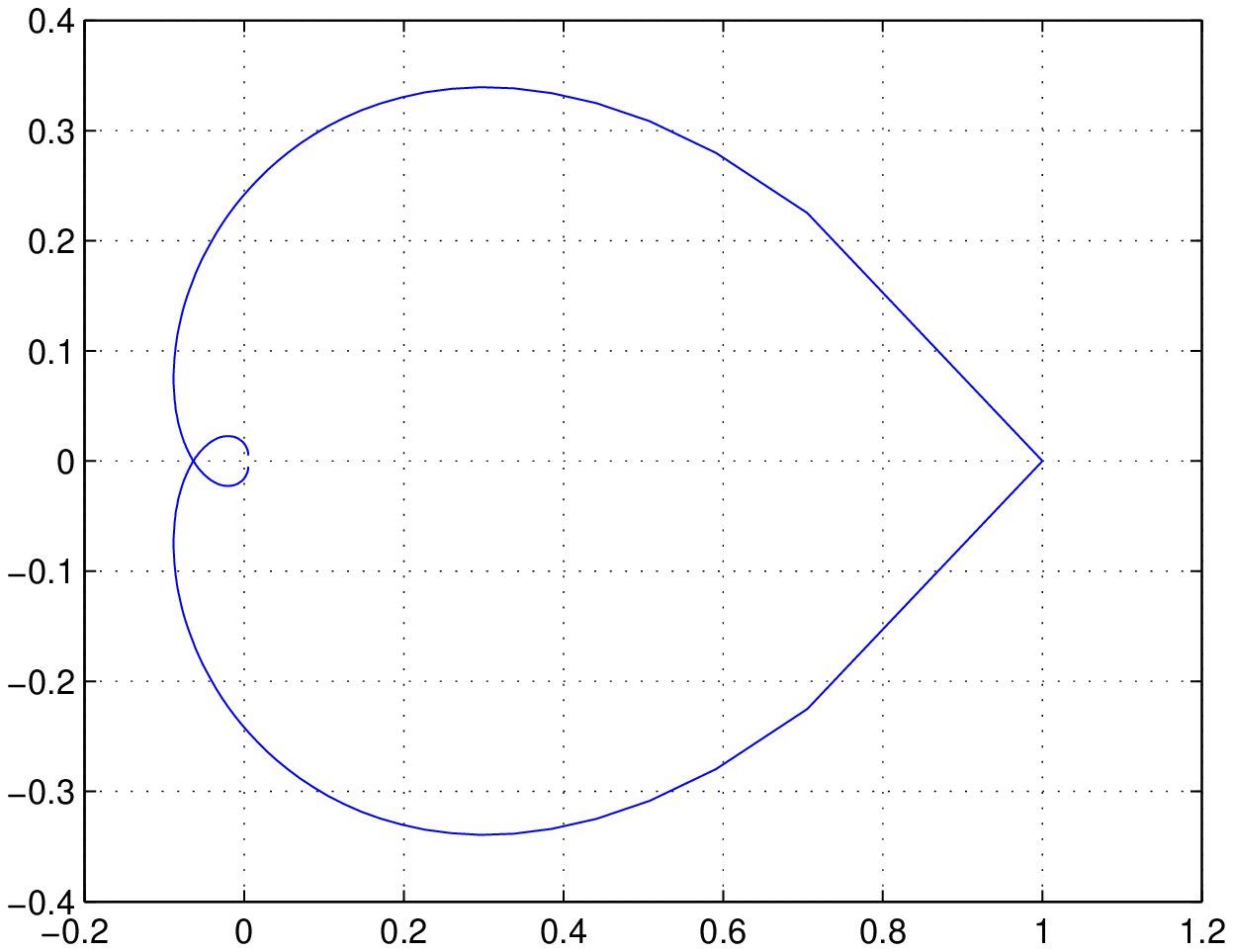}
\includegraphics[width=6.5cm]{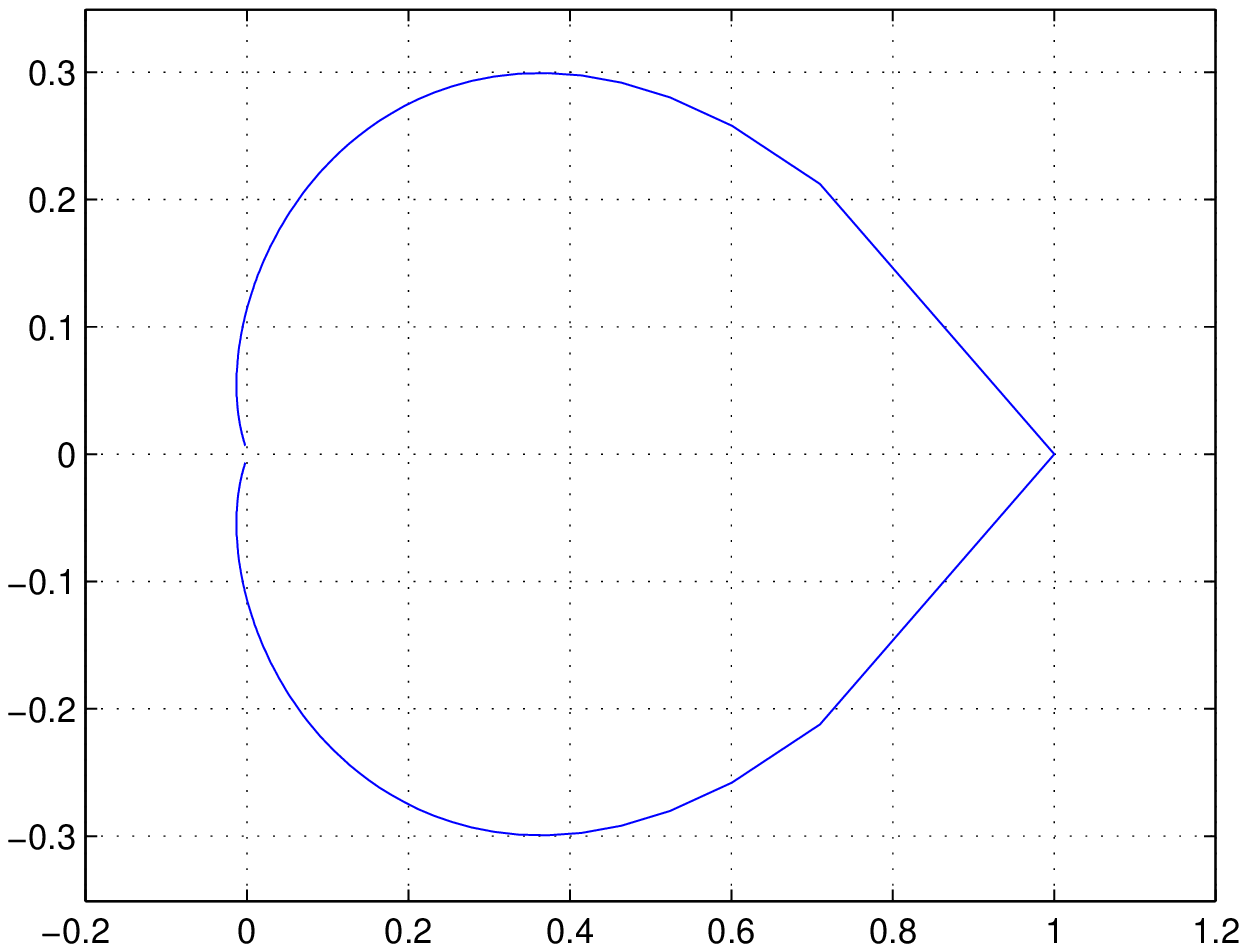}\\
\hspace{-1cm}
\includegraphics[width=6.5cm]{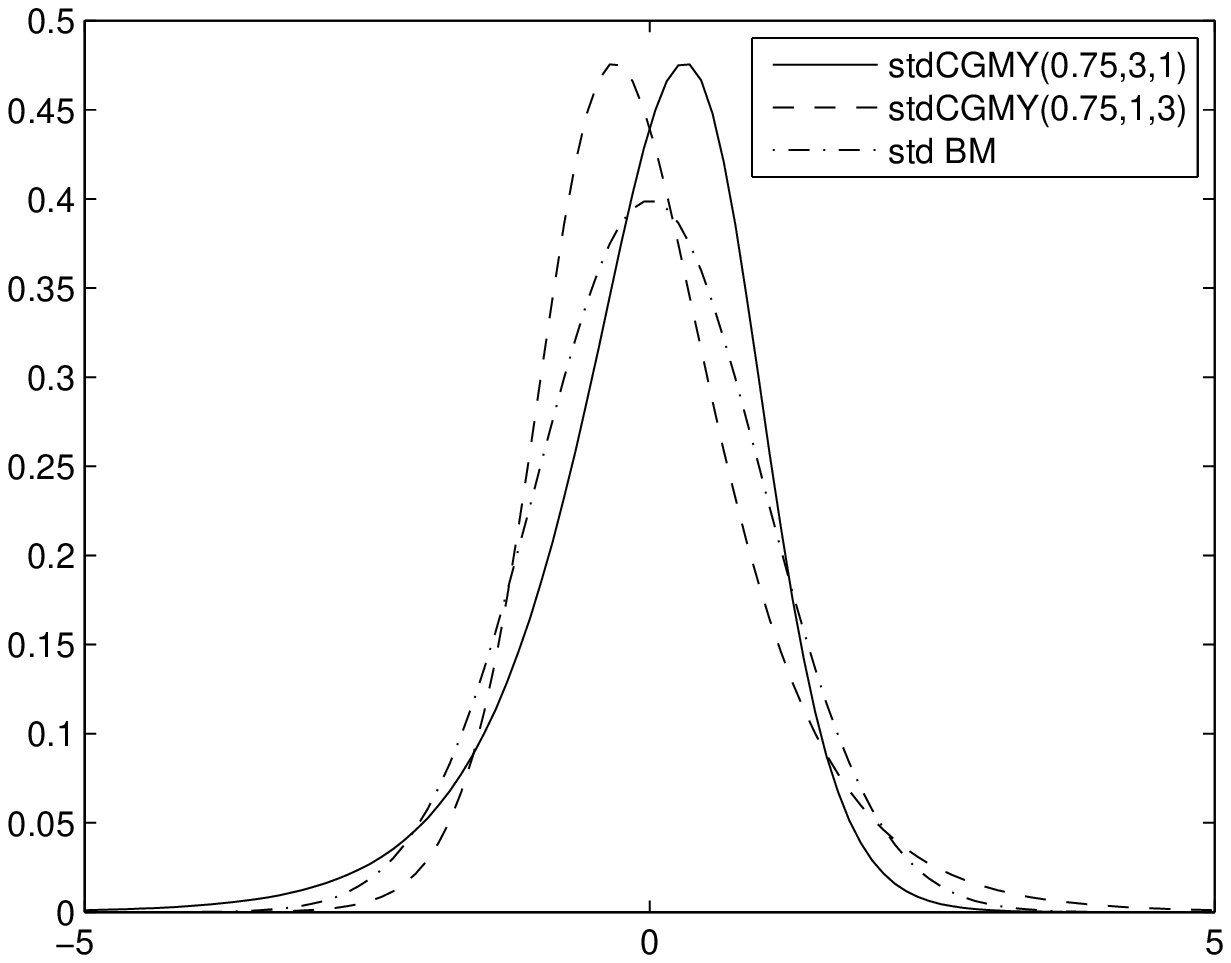}
\includegraphics[width=6.5cm]{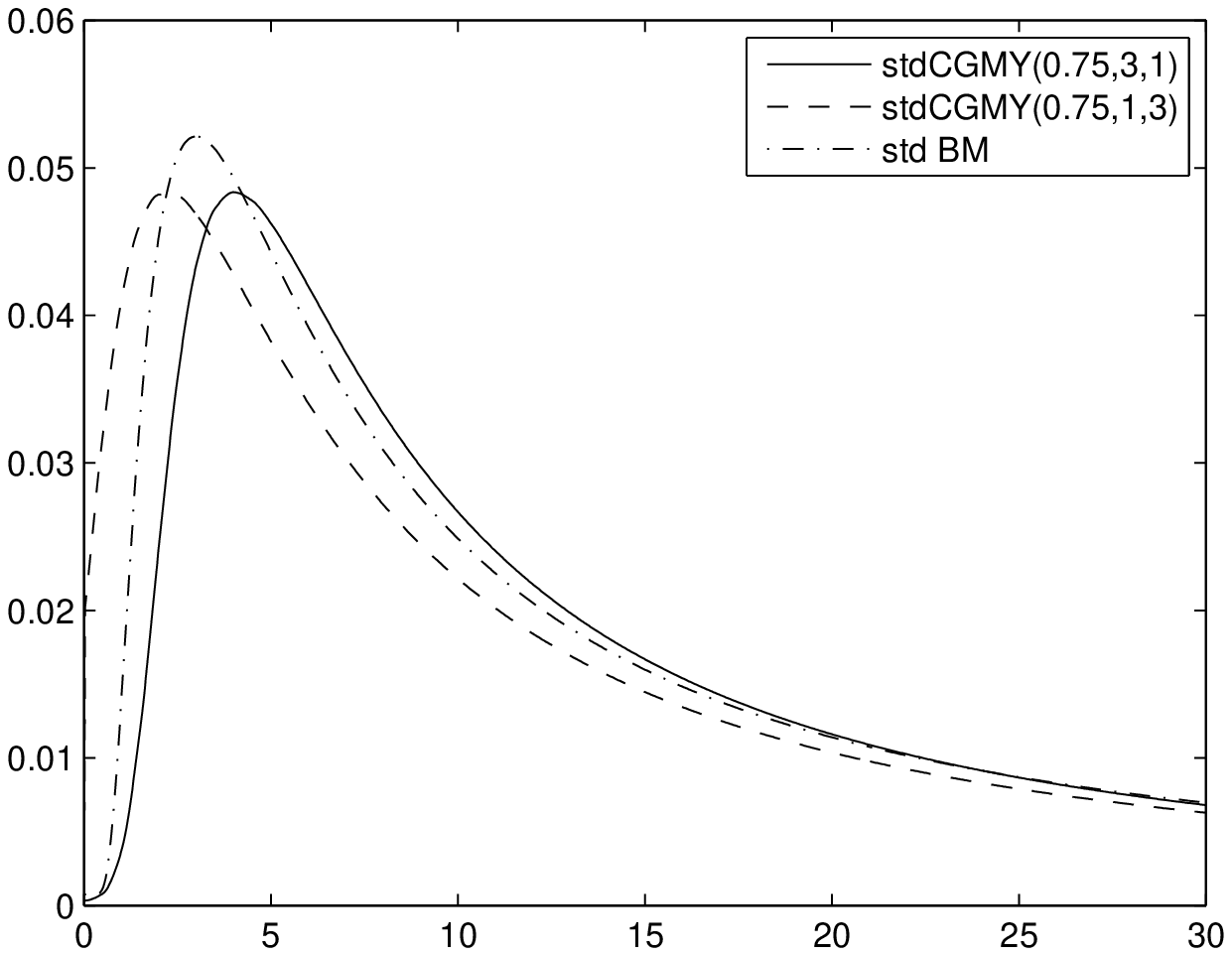}
\end{center}
\caption{\label{fig:firsthittingtime_CGMY}Function $\eta(u)$'s and ch.F's for first passage time of standard CGMY processes. The upper left is the function $\eta(u)$ for stdCGMY$(0.75, 3, 1)$ and the upper right is for stdCGMY$(0.75, 1, 3)$. The middle left is the characteristic function for the stdCGMY$(0.75, 3, 1)$ and the middle right is for stdCGMY$(0.75, 1, 3)$. Pdf's of standard CGMY distributions are on the bottom left. The pdfs of the first passage time of standard CGMY processes are on the bottom right.}
\end{figure}
\clearpage

\begin{figure}[t]
\hspace{-1cm}
\includegraphics[width=7cm]{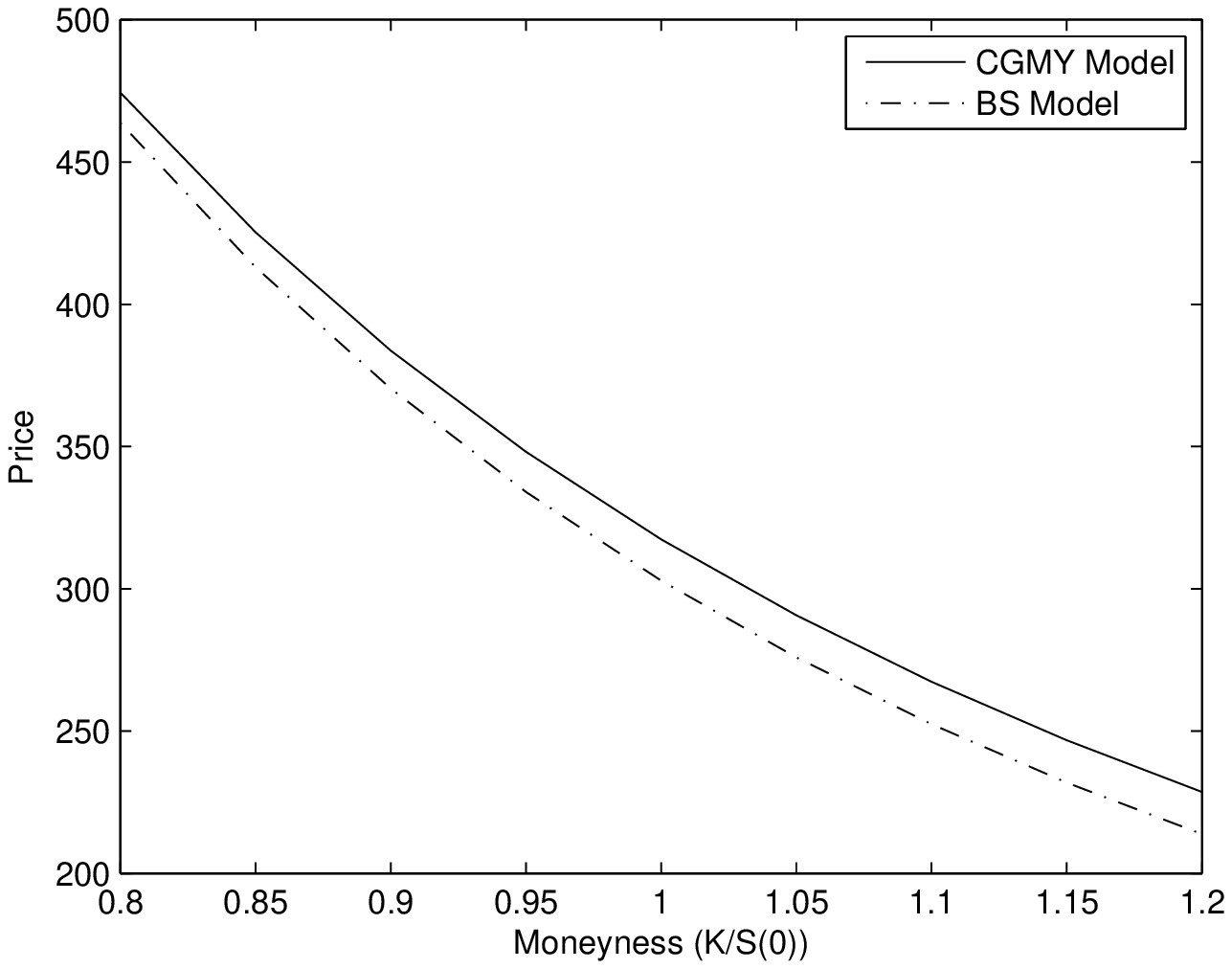}
\includegraphics[width=7cm]{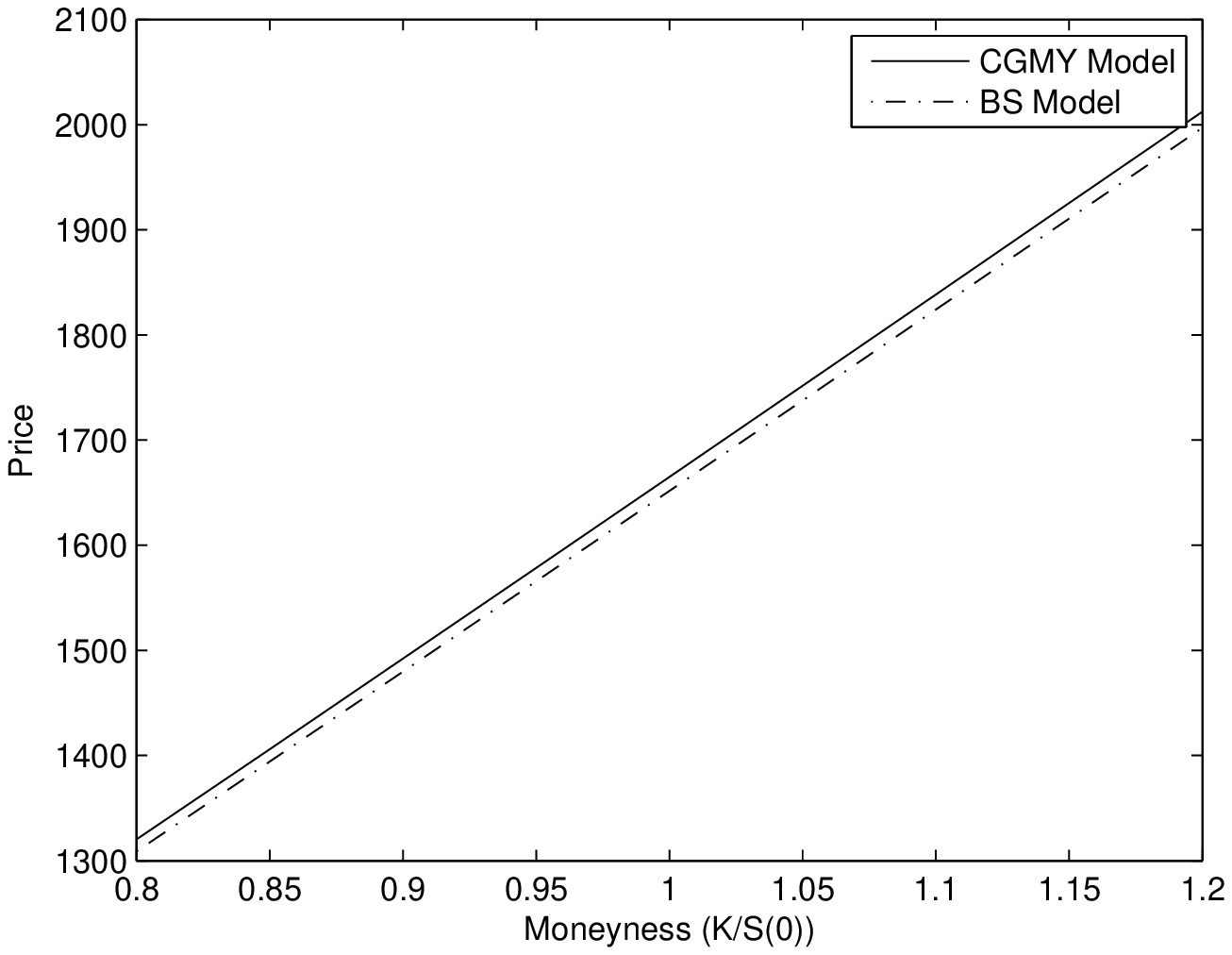}
\caption{\label{fig:PerpetualCallPutPrices}Perpetual call (left) and put (right) prices under the CGMY market model (Solid Curve) and BS model (Dash-dot Curve). }
\end{figure}
\clearpage

\begin{figure}[t]
\hspace{-1cm}
\includegraphics[width=7cm]{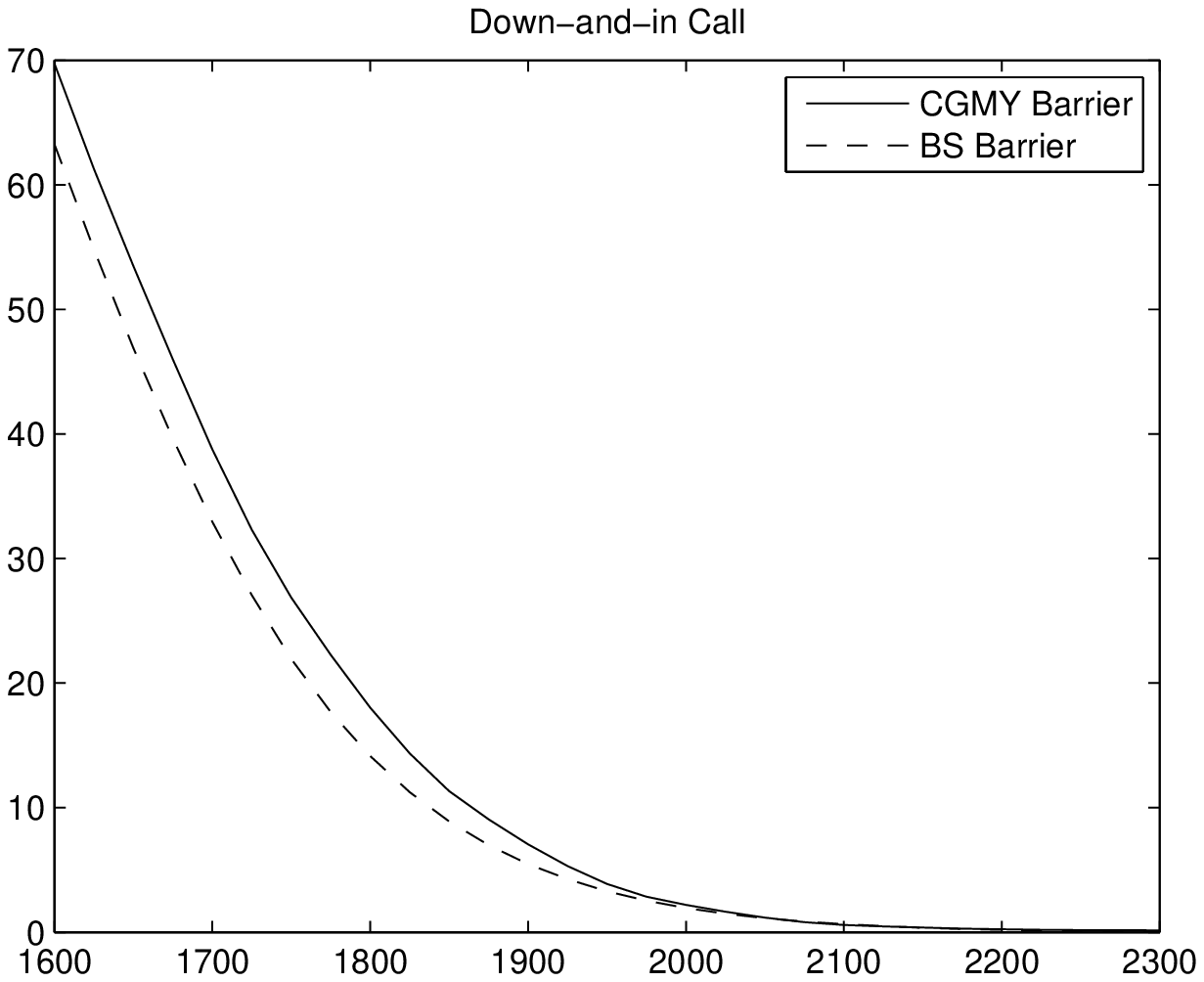}
\includegraphics[width=7cm]{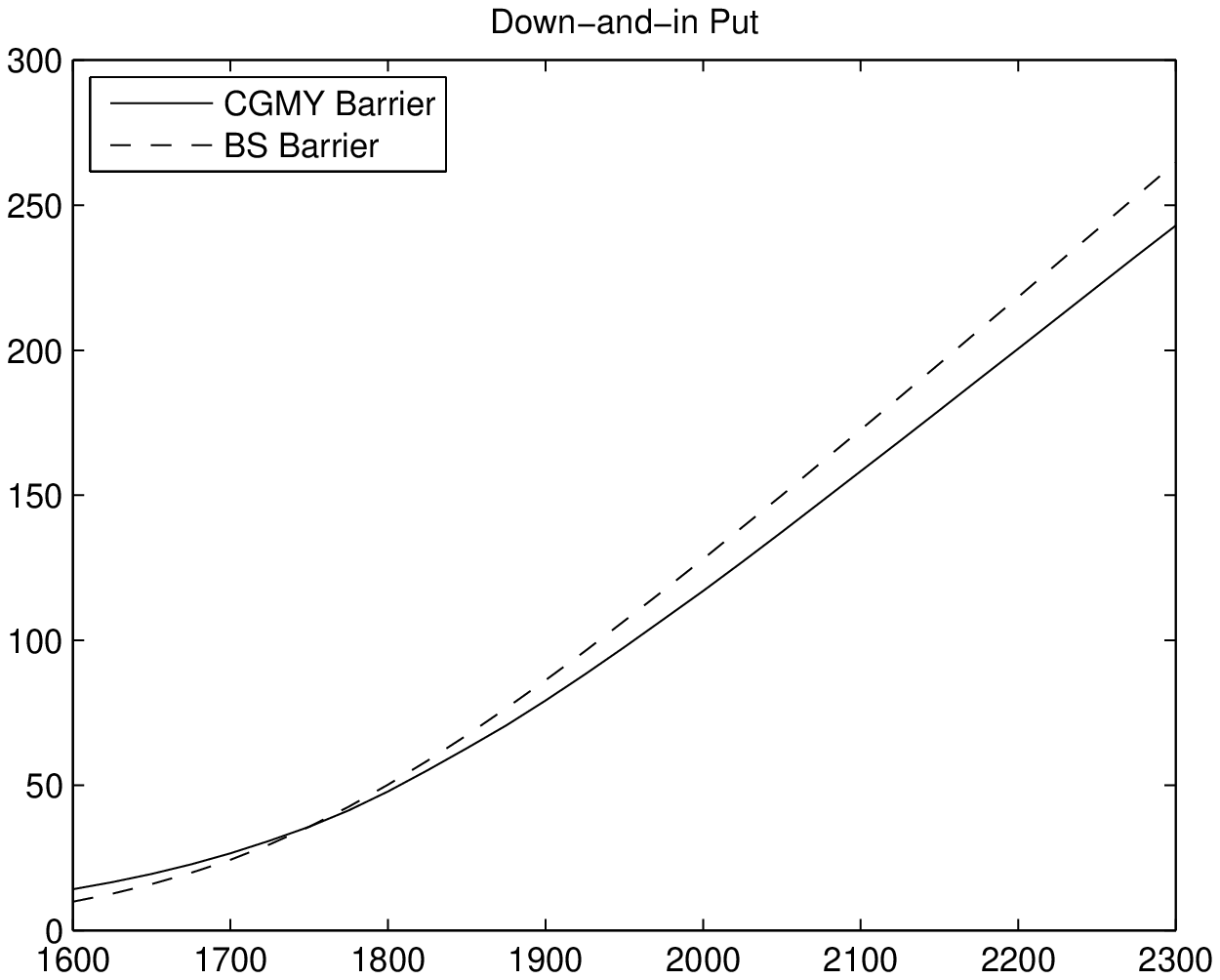}
\caption{\label{fig:DownAndInCallPutPrices}Down-and-in call (left) and put (right) with barrier level 1750 prices, where current underlying index price $S(0) = 1968.89$, and time to maturity is $T$ = 1 year. The solid curves are call/put prices of CGMY model, and the dashed curves are of BS model.}
\end{figure}

\begin{figure}[t]
\hspace{-1cm}
\includegraphics[width=7cm]{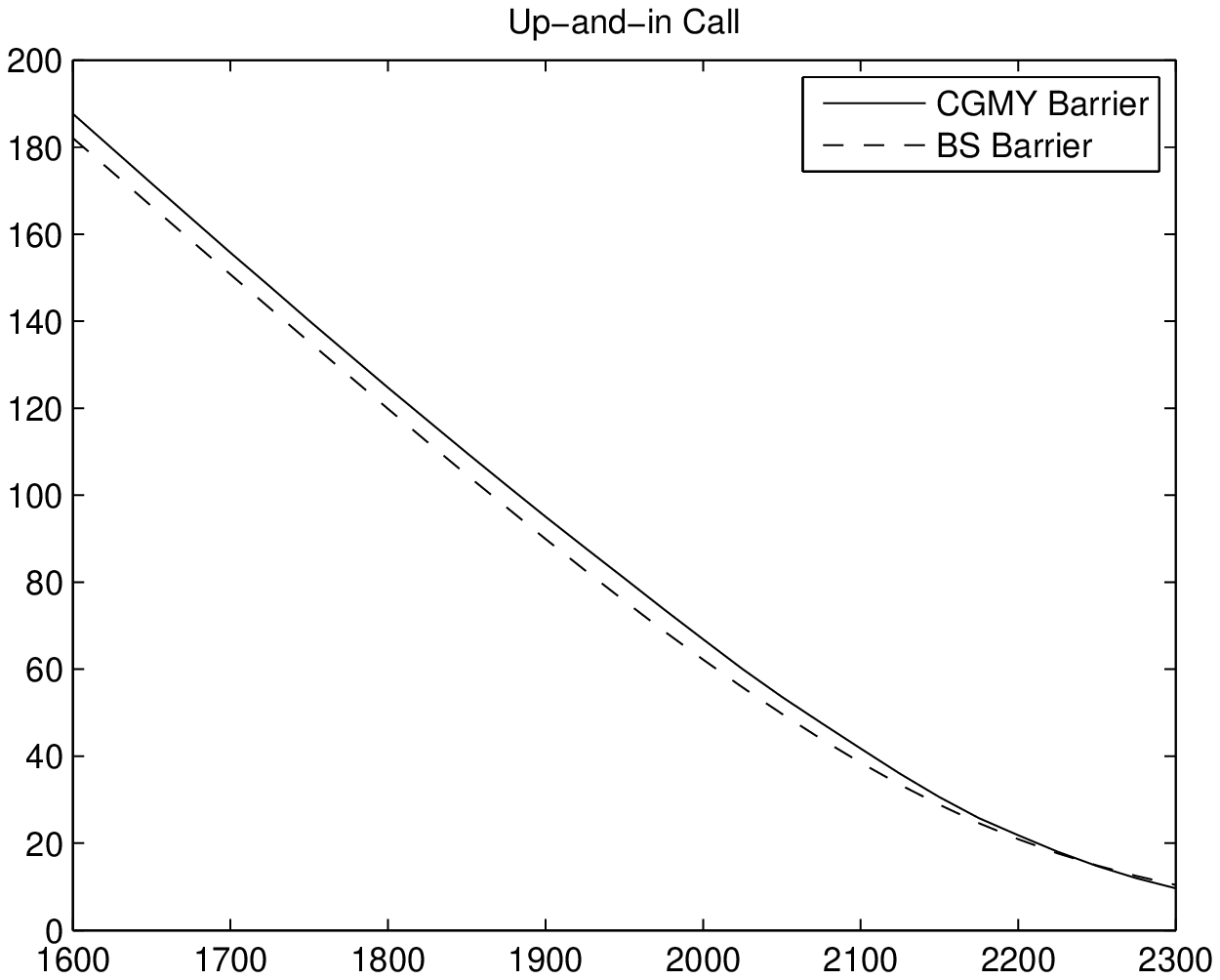}
\includegraphics[width=7cm]{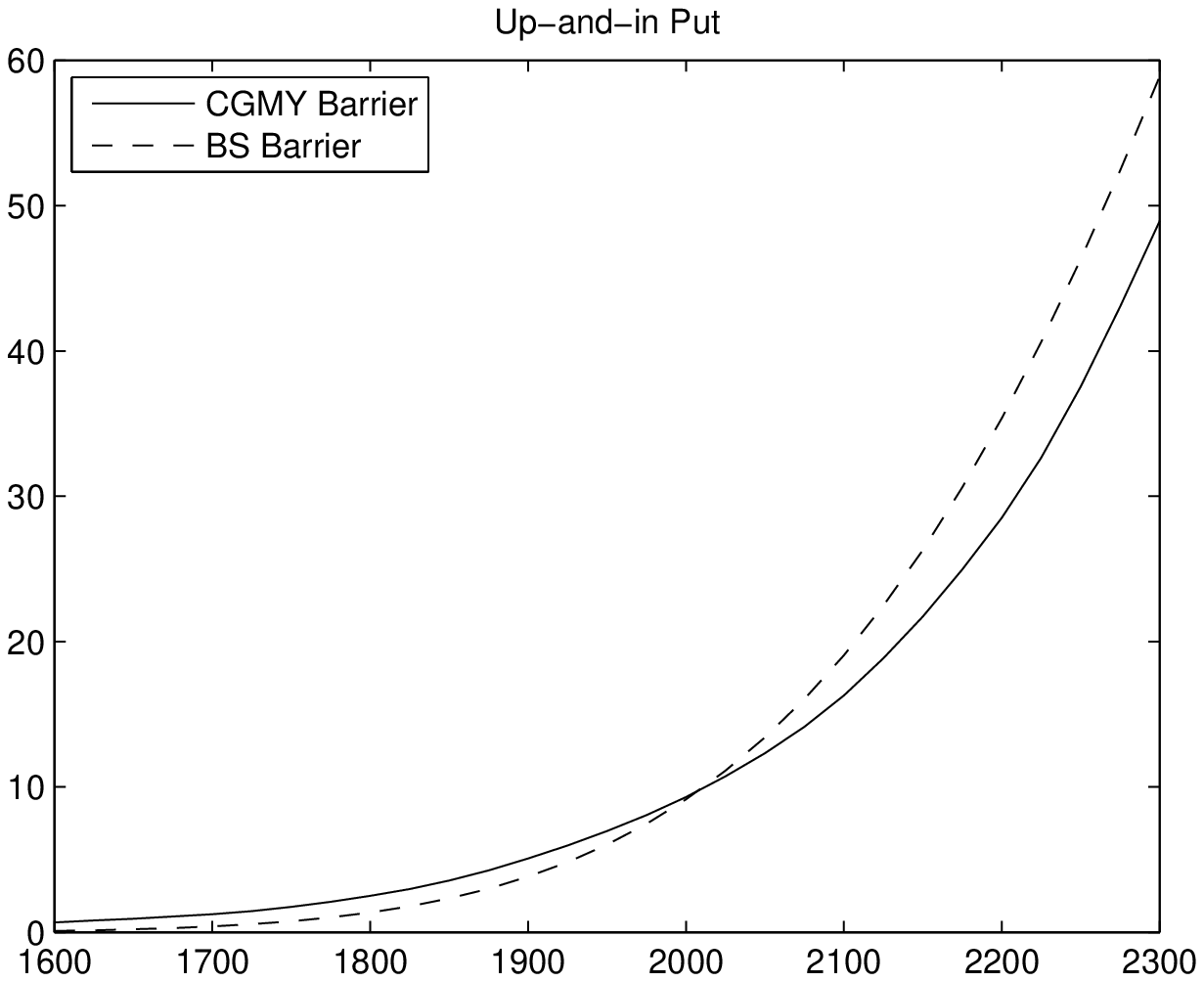}
\caption{\label{fig:UpAndInCallPutPrices}Up-and-in call (left) and put (right) prices where the barrier level is $B = 2200$, current underlying index price $S(0) = 1968.89$, and time to maturity is $T$ = 1 year. The solid curves are call/put prices of CGMY model, and the dashed curves are of BS model.}
\end{figure}

\clearpage
\section*{Appendix}
As appendix, we discuss perpetual American option pricing and barrier option pricing under the \levy~market model.
\subsection*{Perpetual American Option}
The perpetual call and put option price on \levy~model can be obtained by the martingale method introduced in \cite{GerberShiu:1994}. In this section, we just follow the martingale method for the \levy~market price model.
We consider a perpetual American call option with strike price $K$. If the option holder exercise the call at a time $T$, then the holder obtain $(S(T)-K)^+$ where $x^+ = \max\{0,x\}$. 
Let $L$ be a real number with $L\ge K$. The holder will exercise the call when the asset price first become greater than or equal to the level $L$. We define the first passage time
\[
\tau(l) = \inf\{t\ge0|S(t)\ge L\} = \inf\{t\ge0|X(t)\ge l\}, 
\]
where $l=\log(L/S(0))>0$. Then the current value of the perpetual American call is
\[
C_\text{perpetual}=\max_{L\ge K} E[e^{-r\tau(l)}(S(\tau(l))-K)^+].
\]
Let
\[
C(L) = E[e^{-r\tau(l)}(S(\tau(l))-K)^+]=(L-K)E[e^{-r\tau(l)}]
\]
which is the Laplace transform of $\tau(l)$. Applying Lemma \ref{Lemma:ChfTau}, we can obtain the Laplace transform as 
\[
E\left[e^{-r\tau(l)}\right] = \phi_{\tau(l)}(ir) =e^{-l\eta^+(ir)},
\]
where $\eta^+(ir)$ is the value satisfying \eqref{eq:martingalecondition2} and \eqref{eq:chf_tau} for $l>0$ and $u=ir$.
Hence we have
\[
C(L) = (L-K)e^{-l\eta^+(ir)}=(L-K)\left(\frac{S(0)}{L}\right)^{\eta^+(ir)}.
\]
By solving
\[
\frac{\partial C}{\partial L}(L^+) = 0,
\]
we find the optimal value $L^+$ 
\[
L^+ = \frac{\eta^+(ir)K}{\eta^+(ir)-1}.
\]
Hence, we obtain the maximum value
\[
C(L^+)=\frac{K}{\eta^+(ir)-1}\left(\frac{S(0)(\eta^+(ir)-1)}{K\eta^+(ir)}\right)^{\eta^+(ir)}.
\]
If $L^+< S(0)$ then the call is immediately exercised so we have price $S(0)-K$. Therefore the perpetual call price is equal to 
\[
C_\text{perpetual}=\begin{cases}
\displaystyle \frac{K}{\eta^+(ir)-1}\left(\frac{S(0)(\eta^+(ir)-1)}{K\eta^+(ir)}\right)^{\eta^+(ir)} &\text{ if } S(0)\le L^+ \\
S(0)-K &\text{ if } S(0)>L^+
\end{cases}.
\]

We consider a perpetual American put option with strike price $K$. If the option holder exercise the put at a time $T$, then the holder obtain $(K-S(T))^+$. 
Let $L$ be a real number with $0<L\le K$.
The holder will exercise the put when the asset price first become less than or equal to the level $L$. We define the first passage time
\[
\tau(l) = \inf\{t\ge0|S(t)\le L\} = \inf\{t\ge0|X(t)\le l\}  
\]
where $l=\log(L/S(0))<0$. Then the current value of the put is
\[
P_\text{perpetual} = \max_{0<L\le K} E[e^{-r\tau(l)}(K-S(\tau(l)))^+]
\]
which is the Laplace transform of $\tau(l)$.
For the same arguments as the call option case, we find the optimal value $L^-$ 
\[
L^- = \frac{\eta^-(ir)K}{\eta^-(ir)-1},
\]
where  $\eta^-(ir)$ is the value satisfying \eqref{eq:martingalecondition2} and \eqref{eq:chf_tau} for $l<0$ and $u=ir$.
Hence the perpetual put price is equal to 
\[
P_\text{perpetual}=\begin{cases}
\displaystyle \frac{K}{1-\eta^-(ir)}\left(\frac{S(0)(\eta^-(ir)-1)}{K\eta^-(ir)}\right)^{\eta^-(ir)} &\text{ if } S(0)\ge L^- \\
K-S(0) &\text{ if } S(0)<L^-
\end{cases}.
\]

\subsection*{Barrier Option}
Let $Pi$ be the payoff function of European options. For example, the European call and put options with strike price $K$ are given by $\Pi(S(T))=(S(T)-K)^+$ and $\Pi(S(T))=(K-S(T))^+$, respectively. 
The knock-in barrier option with the barrier level $B$, time to maturity $T$ is priced by the following equation
\[
V_{i}=e^{-rT}E\left[\Pi(S(T))1_{\tau(l)<T}\right]
\]
where $l = \log(B/S(0))$. Note that $l<0$ for the down-and-in barrier option and $l>0$ for the up-and-in barrier option. 
Since we have 
\[
\Pi(S(T)) = \Pi(S(T))1_{\tau(l)<T}+\Pi(S(T))1_{\tau(l)\ge T},
\]
the knock-out barrier option price can be obtained by the following equation
\[
V_{o}=e^{-rT}E\left[\Pi(S(T))1_{\tau(l)\ge T}\right] = e^{-rT}\left(E\left[\Pi(S(T))\right] - E\left[\Pi(S(T))1_{\tau(l)< T}\right] \right) = V - V_{i},
\]
where $V = e^{-rT}E\left[\Pi(S(T))\right]$. Note that $l<0$ for the down-and-out barrier option and $l>0$ for the up-and-out barrier option.

\noindent\textit{Case 1:} $\Pi(S(T)) = \Pi(S(T))1_{\tau(l)<T}$\\
If $\Pi(S(T)) = \Pi(S(T))1_{\tau(l)<T}$ then the barrier option price is the same as option prices without the barrier:
\[
V_{i}=e^{-rT}E\left[\Pi(S(T))1_{\tau(l)<T}\right] =e^{-rT}E\left[\Pi(S(T))\right].
\]
For example (1) up-and-in call option with $K>B$, we have $(S(T)-K)^+=(S(T)-K)^+1_{\tau(l)<T}$, and (2) down-and-in put option with $K<B$, we have $(K-S(T))^+=(K-S(T))^+1_{\tau(l)<T}$. 

\noindent\textit{Case 2:} $\Pi(S(T)) \neq \Pi(S(T))1_{\tau(l)<T}$\\
If $\Pi(S(T))\neq \Pi(S(T))1_{\tau(l)<T}$, we have
\begin{align*}
V_{i}&=e^{-rT}E\left[\Pi(S(T))1_{\tau(l)<T}\right] \\
&=e^{-rT}E\left[E\left[\Pi(S(0)e^{X(T)-X(\tau(l))+X(\tau(l))})1_{\tau(l)<T}|\tau(l)\right]\right] \\
&=e^{-rT}E\left[E\left[\Pi(S(0)e^le^{X(T-\tau(l))})1_{\tau(l)<T}|\tau(l)\right]\right] \\
&=e^{-rT}\int_0^T E\left[\Pi(S(0)e^le^{X(T-t)})1_{t<T}\right]f_{\tau(l)}(t)dt
\end{align*}
where $f_{\tau(l)}$ is the pdf of $\tau(l)$.
Since we have $S(0)e^l = B$ and
\[
f_{\tau(l)}(t)=\frac{1}{2\pi}\int_{-\infty}^\infty e^{-ivt}\phi_{\tau(l)}(v)dv,
\]
the $c_{i}$ becomes
\begin{align}
\nonumber V_{i}&=\frac{e^{-rT}}{2\pi}\int_0^T E\left[\Pi(Be^{X(T-t)})\right]\int_{-\infty}^\infty e^{-ivt}\phi_{\tau(l)}(v)dv\, dt \\
&=\frac{e^{-rT}}{2\pi}\int_{-\infty}^\infty \int_0^T E\left[\Pi(Be^{X(T-t)})\right] e^{-ivt}dt\phi_{\tau(l)}(v)dv. \label{beforputcallparity}
\end{align}
By European option pricing formula using Fourier transform (See \cite{CarrMadan:1999}, \cite{Lewis:2001} and \cite{RachevKimBianchiFabozzi:2011a}),
we have
\[
E\left[\Pi(Be^{X(T-t)})\right]=\frac{1}{2\pi}\int_{-\infty}^\infty B^{i(u+i\rho)}e^{(T-t)\psi_X(u+i\rho)}\hat\Pi(u+i\rho)du,
\]
where $\hat\Pi(z)=\int_{-\infty}^\infty e^{-izx}\Pi(e^x)dx$ for complex number $z$ and $\rho$ is a real constant such that $\psi_X(u+i\rho)$ and $\hat\Pi(u+i\rho)$ are well defined for all $u\in\R$.  Hence we have
\begin{align*}
V_{i}&=\frac{e^{-rT}}{2\pi}\int_{-\infty}^\infty \int_0^T \frac{1}{2\pi}\int_{-\infty}^\infty B^{i(u+i\rho)}e^{(T-t)\psi_X(u+i\rho)}\hat\Pi(u+i\rho)du\, e^{-ivt}dt\, \phi_{\tau(l)}(v)dv\\
&=\frac{e^{-rT}}{(2\pi)^2}\int_{-\infty}^\infty \int_{-\infty}^\infty B^{i(u+i\rho)}\hat\Pi(u+i\rho)\int_0^T e^{(T-t)\psi_X(u+i\rho)} e^{-ivt}dt\, du\, \phi_{\tau(l)}(v)dv\\
&=\frac{e^{-rT}}{(2\pi)^2}\int_{-\infty}^\infty B^{i(u+i\rho)}\hat\Pi(u+i\rho)\int_{-\infty}^\infty \frac{e^{T\psi_X(u+i\rho)}-e^{-ivT}}{\psi_X(u+i\rho)+iv}\phi_{\tau(l)}(v)dv\, du\\
\end{align*}
Let 
\[
H(u) = \int_{-\infty}^\infty \frac{e^{T\psi_X(u+i\rho)}-e^{-ivT}}{\psi_X(u+i\rho)+iv}\phi_{\tau(l)}(v)dv
\]
then 
\[
V_i=\frac{e^{-rT}}{(2\pi)^2}\int_{-\infty}^\infty B^{i(u+i\rho)}\hat\Pi(u+i\rho)H(u) du.
\]

\noindent\textit{European call and put options}\\
For the call option payoff $\Pi(S(T))=(S(T)-K)^+$, we have
\[
\hat\Pi(u+i\rho) = \int_{\log K}^\infty e^{-i(u+i\rho)x}(e^x-K)dx = \frac{K^{\rho+1-iu}}{(\rho-iu)(\rho+1-iu)}, ~~~\rho<-1
\]
and for the put option payoff $\Pi(S(T))=(K-S(T))^+$, we have
\[
\hat\Pi(u+i\rho) = \int^{\log K}_{-\infty} e^{-i(u+i\rho)x}(K-e^x)dx = \frac{K^{\rho+1-iu}}{(\rho-iu)(\rho+1-iu)}, ~~~\rho>0.
\]

The down-and-in call option price ($c_{di}$) and up-and-in put option price ($p_{ui}$) are always in Case 2.
Therefore, we have their prices as
\[
c_{di} = \frac{e^{-rT}K^{1+\rho}}{(2\pi)^2 B^\rho}\int_{-\infty}^\infty \left(\frac{B}{K}\right)^{iu} \left(\frac{H(u)}{(\rho-iu)(1+\rho-iu)}\right)du
, ~~~\rho<-1
\]
and
\[
p_{ui} = \frac{e^{-rT}K^{1+\rho}}{(2\pi)^2 B^\rho}\int_{-\infty}^\infty \left(\frac{B}{K}\right)^{iu} \left(\frac{H(u)}{(\rho-iu)(1+\rho-iu)}\right)du
, ~~~\rho>0
\]

For the up-and-in call option price ($c_{ui}$) and the down-and-in put option price ($p_{di}$), we consider Case 1, and finally obtain  
\[
c_{ui} = \begin{cases}
\frac{e^{-rT}K^{1+\rho}}{(2\pi)^2 B^\rho}\int_{-\infty}^\infty \left(\frac{B}{K}\right)^{iu} \left(\frac{H(u)}{(\rho-iu)(1+\rho-iu)}\right)du & \text{ if } K\le B \\
\frac{e^{-rT}K^{1+\rho}}{2\pi S(0)^\rho}\int_{-\infty}^\infty \left(\frac{S(0)}{K}\right)^{iu} \left(\frac{\phi_{X(T-t)}(u+i\rho)}{(\rho-iu)(1+\rho-iu)}\right)du & \text{ if } K> B
\end{cases}
, ~~~\rho<-1.
\]
and
\[
p_{di} = \begin{cases}
\frac{e^{-rT}K^{1+\rho}}{(2\pi)^2 B^\rho}\int_{-\infty}^\infty \left(\frac{B}{K}\right)^{iu} \left(\frac{H(u)}{(\rho-iu)(1+\rho-iu)}\right)du & \text{ if } K\ge B \\
\frac{e^{-rT}K^{1+\rho}}{2\pi S(0)^\rho}\int_{-\infty}^\infty \left(\frac{S(0)}{K}\right)^{iu} \left(\frac{\phi_{X(T-t)}(u+i\rho)}{(\rho-iu)(1+\rho-iu)}\right)du & \text{ if } K< B
\end{cases}
, ~~~\rho>0.
\]

The up-and-out and down-and-out calls corresponding to the up-and-in and down-and-in calls above are priced by $c_{uo} = c - c_{ui}$, and $c_{do} = c - c_{di}$, respectively,
where $c$ is the vanilla call option price with  strike price $K$, and time to maturity $T$.
The up-and-out and down-and-out puts corresponding to the up-and-in and down-and-in puts above are priced by $p_{uo} = p - p_{ui}$, and $p_{do} = p - p_{di}$, respectively,
where $p$ is the vanilla put option price with  strike price $K$, and time to maturity $T$.

\end{document}